\documentclass[10pt,a4paper]{article}

\usepackage{amsmath}
\usepackage{amsfonts}
\usepackage{amsthm}
\usepackage{amssymb}
\usepackage{ascmac}
\usepackage{authblk}
\usepackage[dvipdfmx]{graphicx}
\newtheorem{theorem}{Theorem}

\newtheorem{lemma}{Lemma}
\newtheorem{proposition}{Proposition}

\newcommand{\bvec}[1]{\mbox{\boldmath $#1$}}

\title{Discrete-time analysis for the integrable discrete Toda equations and the discrete Lotka-Volterra system}
\author[1]{Masato Shinjo \thanks{mshinjo@amp.i.kyoto-u.ac.jp}}
\author[1]{Yoshimasa Nakamura}
\author[2]{Masashi Iwasaki}
\author[3]{Koichi Kondo}
\affil[1]{
              Graduate School of Informatics, Kyoto University, 
              Yoshida-Hommachi, Sakyo-ku, Kyoto 606-8501, Japan}
\affil[2]{
              Faculty of Life and Environmental Sciences, 
              Kyoto Prefectural University,
              1-5, Nakaragi-cho, Shimogamo, Sakyo-ku, 
              Kyoto 606-8522, Japan}
\affil[3]{
              Graduate School of Science and Engineering,
              Doshisha University,
              1-3, Tatara miyakodani, Kyotanabe, 
              Kyoto 610-0394, Japan}              
   
\date{}
\begin{document}
\maketitle

\begin{abstract}%
The discrete autonomous/non-autonomous Toda equations and the discrete Lotka-Volterra system are important  integrable discrete systems in fields such as mathematical physics, mathematical biology and statistical physics.
They also have applications to numerical linear algebra. 
In this paper, we first simultaneously obtain their general solutions.
Then, we show the asymptotic behavior of the solutions for any initial values as the discrete-time variables go to infinity. 
Our two main techniques for understanding the distinct integrable systems are to introduce two types of discrete-time variables and to examine properties of a restricted infinite sequence, its associated determinants and polynomials.
\end{abstract}

\section{Introduction}\label{intro}
Many integrable systems first appeared in mathematical physics and mathematical  biology for real-life problems. 
Time discretizations of such integrable systems can be useful for understanding physical and biological phenomena numerically. 
Discrete-time evolutions in integrable discrete systems also contribute to scientific computing as key components of numerical algorithms.
\par
The Toda equation is one of the most famous integrable systems, and describes the motion governed by a nonlinear spring \cite{Toda1981},
but studies have branched into, for example, nonlinear electric circuits \cite{Hirota1973}, explicit soliton solutions \cite{Nakamura1997}, and the connection with simple Lie algebra \cite{Bogoyavlensky1976}. 
A time-discretization of the Toda equation is
\begin{align}\label{adToda}
\left\{
\begin{array}{l}
\displaystyle q_{k}^{(s+1)}+e_{k-1}^{(s+1)}=q_{k}^{(s)}+e_{k}^{(s)},\quad k=1,2,\dots,m,\\[2pt]
\displaystyle q_{k}^{(s+1)}e_{k}^{(s+1)}=q_{k+1}^{(s)}e_{k}^{(s)},\quad k=1,2,\dots,m-1,\\[2pt]
\displaystyle e_{0}^{(s)}:=0,\quad e_{m}^{(s)}:=0,\\[2pt]
s=0,1,\dots,
\end{array}
\right.
\end{align}
where the subscripts $k$ and superscripts $s$ with parentheses are discrete-space and discrete-time variables, respectively.
Interestingly, the Toda equation and its time discretization both have relationships to well-known algorithms for computing eigenvalues of tridiagonal matrices. 
The time evolution in the Toda equation corresponds to the $1$-step of the $QR$ algorithm for tridiagonal matrix exponentials \cite{Symes1982}. 
The discrete Toda (dToda) equation \eqref{adToda} is equivalent to the recursion formula that generates the similarity $LR$ transformations of tridiagonal matrices in the quotient-difference (qd) algorithm \cite{Rutishauser1990}. 
The qd algorithm is also used to compute singular values of bidiagonal matrices.
\par
In the theory of orthogonal polynomials, the compatibility condition of a transformation for discrete-time evolution and its inverse yields the  following dToda equation with arbitrary constants $\mu^{(t)}$ and $\mu^{(t+1)}$ \cite{Hirota1997}:
\begin{align}\label{nadToda}
\left\{
\begin{array}{l}
\displaystyle Q_{k}^{(t+1)}+E_{k-1}^{(t+1)}+\mu^{(t+1)}=Q_{k}^{(t)}+E_{k}^{(t)}+\mu^{(t)}, 
\quad k=1,2,\dots,m,\\[2pt]
\displaystyle Q_{k}^{(t+1)}E_{k}^{(t+1)}=Q_{k+1}^{(t)}E_{k}^{(t)},\quad k=1,2,\dots,m-1,\\[2pt]
\displaystyle E_{0}^{(t)}:=0,\quad E_{m}^{(t)}:=0,\\[2pt]
 t=0,1,\dots.
\end{array}
\right.
\end{align}
We consider $t$ appearing in the superscripts in \eqref{nadToda} to be a different discrete-time variable from $s$ in the dToda equation \eqref{adToda}. 
Since $\mu^{(t)}$ and $\mu^{(t+1)}$ depend on the independent variable $t$, we refer to \eqref{nadToda} as the non-autonomous dToda equation, and refer to \eqref{adToda} as the autonomous dToda equation. 
If we set $s=t$ and $\mu^{(t)}=\mu^{(t+1)}=0$, then the non-autonomous dToda equation \eqref{nadToda} is equivalent to the autonomous dToda equation \eqref{adToda}.
The non-autonomous dToda equation \eqref{nadToda} also gives the $LR$ transformations with implicit shifts of tridiagonal matrices. 
In other words, \eqref{nadToda} can be considered a recursion formula of the shifted qd algorithm that may converge faster than the original algorithm. 
The constants $\mu^{(t)}$ and $\mu^{(t+1)}$ then correspond to the shifts that bring about the convergence acceleration.
\par
The integrable Lotka-Volterra (LV) system is a simple biological model describing a predator-prey interaction of several species. 
Aside from the discretization of the Toda equation, the discrete LV (dLV) system has discretization parameters $\delta^{(t)}$ and $\delta^{(t+1)}$ as follows:
\begin{align}\label{dLV}
\left\{
\begin{array}{l}
\displaystyle 
u_k^{(t+1)}(1+\delta^{(t+1)}u_{k-1}^{(t+1)})
=u_k^{(t)}(1+\delta^{(t)}u_{k+1}^{(t)}),\quad
k=1,2,\dots,2m-1,\\[2pt]
\displaystyle 
u_0^{(t)}:=0,\quad u_{2m}^{(t)}:=0,\quad\\[2pt]
 t=0,1,\dots.
\end{array}
\right.
\end{align}
Here, $u_k^{(t)}$ indicates density of the $k$th species at the discrete-time $t$.
The dLV system \eqref{dLV} also generates the $LR$ and shifted $LR$ transformations of 
positive-definite tridiagonal matrices for computing singular values of bidiagonal matrices \cite{Iwasaki2002,Iwasaki2004,Tsujimoto2001}.
\par
Several observations have been made regarding the asymptotic behavior of solutions 
to the autonomous dToda equation \eqref{adToda} and the dLV system \eqref{dLV}. 
The latter, however, has only been analyzed with positive initial values. 
These asymptotic behaviors have been individually shown using distinct frameworks in separate books \cite{Henrici,Rutishauser1990} and separate papers \cite{Iwasaki2002,Iwasaki2004,Tsujimoto2001}. 
Moreover, to the best of our knowledge, asymptotic analysis for the non-autonomous dToda equation \eqref{nadToda} has not been reported yet. 
In this paper, based on only an infinite sequence with respect to two types of discrete-time variables $s$ and $t$,
we report on the asymptotic behavior of all three systems.
\par
The remainder of this paper is organized as follows.
In Section~\ref{sec:2}, we first introduce an infinite sequence with respect to two types of discrete-time variables $s$ and $t$.
The infinite sequence is also associated with matrix eigenvalues. 
We then show properties of the Hankel determinants associated with the infinite sequence. 
In Sections~\ref{sec:3} and~\ref{sec:31}, we examine relationships between the Hankel determinants, 
their Hankel polynomials, and Hadamard polynomials with respect to $s$ and $t$. 
Referring to \cite{Maeda2013}, 
we simultaneously derive the autonomous dToda equation \eqref{adToda} and the non-autonomous dToda equation \eqref{nadToda}.
In Section~\ref{sec:4}, we observe eigenvalue problems of square matrices involving the dToda variables, 
and then express the determinant solutions with sufficient degrees of freedom to the dToda equations using the Hankel determinants. 
In Section~\ref{sec:5}, we present asymptotic analysis of the dToda equations by considering asymptotic expansion of the Hankel determinants as $s\rightarrow\infty$ or $t\rightarrow\infty$. 
In Section~\ref{sec:6}, by relating the non-autonomous dToda equation \eqref{nadToda} to the dLV system \eqref{dLV}, 
we describe the determinant solution to the dLV system \eqref{dLV} and its asymptotic behavior as $t\rightarrow\infty$. 
Finally, we conclude in Section~\ref{sec:7}.

\section{Hankel determinants with two discrete-time variables}\label{sec:2}
In this section, we first introduce an infinite sequence $\{f_s^{(t)}\}_{s,t=0}^{\infty}$ with respect to two types of discrete-time variables $s$ and $t$, and then associate it with the Hankel determinants.
Next, we derive identities concerning the Hankel determinants 
under certain restrictions on the infinite sequence $\{f_s^{(t)}\}_{s,t=0}^{\infty}$.
\par
For arbitrary complex $\lambda_1,\lambda_2,\dots,\lambda_m$,
let us introduce an $m$-degree polynomial with respect to complex $z$:
\begin{align}\label{pz}
p(z)=(z-\lambda_1)(z-\lambda_2)\cdots(z-\lambda_m).
\end{align}
The polynomial \eqref{pz} can be regarded as the characteristic polynomial of $m$-by-$m$ matrices 
with eigenvalues $\lambda_1,\lambda_2,\dots,\lambda_m$.
Clearly, the polynomial $p(z)$ can be expanded as
\begin{align}\label{pza}
p(z)=z^m+a_1z^{m-1}+\dots+a_{m-1}z+a_{m},
\end{align}
where $a_1,a_2,\dots,a_n$ are complex constants.
\par
At each $t=0,1,\dots$, let $\{f_{s}^{(t)}\}_{s=0}^{\infty}$ be an infinite sequence 
whose entries satisfy the linear equation with coefficients $a_1,a_2,\dots,a_m$ appearing in \eqref{pza},
\begin{align}\label{fnmt}
f_{s+m}^{(t)}+a_1f_{s+m-1}^{(t)}+\dots+a_mf_{s}^{(t)}=0,\quad s=0,1,\dots,
\end{align}
where $f_0^{(t)},f_{1}^{(t)},\dots,f_{m-1}^{(t)}$ are arbitrary.
Moreover,
let the infinite sequence $\{f_s^{(t)}\}_{s,t=0}^{\infty}$ satisfy
\begin{align}\label{key}
f_{s}^{(t+1)}=f_{s+1}^{(t)}-\mu^{(t)}f_{s}^{(t)}, \quad s,t=0,1,\dots,
\end{align}
where $\{\mu^{(t)}\}_{t=0}^{\infty}$ is an arbitrary constant sequence. \par
We next consider the determinants of symmetric square matrices of order $k$:
\begin{align}\label{Hankel}
H_k^{(s,t)}:=
\left|\begin{array}{cccc}
f_{s}^{(t)}	&f_{s+1}^{(t)}	&\cdots&	f_{s+k-1}^{(t)}\\
f_{s+1}^{(t)}	&f_{s+2}^{(t)}	&\cdots&	f_{s+k}^{(t)}\\
\vdots    &\vdots      	&\vdots&\vdots\\
f_{s+k-1}^{(t)}	&f_{s+k}^{(t)}	&\cdots&	f_{s+2k-2}^{(t)}
\end{array}\right|,
\quad s,t=0,1,\dots,
\end{align}
where $H_{-1}^{(s,t)}:=0$ and $H_{0}^{(s,t)}:=1$.
The determinants $H_k^{(s,t)}$ are known as the Hankel determinants and 
can be used for expressing solutions to several integrable systems.\par
We now derive a proposition concerning the case $k=m+1$ in the Hankel determinants $H_k^{(s,t)}$. 
%
%
\begin{proposition}\label{prop1}
It holds that
\begin{align}
H_{m+1}^{(s,t)}=0,\quad s,t=0,1,\dots.\label{taubc}
\end{align}
\end{proposition}
\begin{proof}
By multiplying the $1$st, $2$nd, $\dots$, $m$th columns in the Hankel determinant $H_{m+1}^{(s,t)}$ 
by $a_m$, $a_{m-1}$, $\dots$, $a_1$, respectively, and by adding these results to the $(m+1)$th column,
we obtain
\begin{align}\label{taum+1st}
H_{m+1}^{(s,t)}=
\left|
\begin{array}{ccccc}
\bvec{f}_{m+1}^{(s,t)}&\bvec{f}_{m+1}^{(s+1,t)}&\cdots&\bvec{f}_{m+1}^{(s+m-1,t)}&
\displaystyle\sum_{\ell=0}^{m}a_{\ell}\bvec{f}_{m+1}^{(s+m-\ell,t)}\\
\end{array}
\right|,
\end{align}
where $a_0:=1$ and  $\bvec{f}_k^{(s,t)}:=(f_{s}^{(t)}, f_{s+1}^{(t)},\dots,f_{s+k-1}^{(t)})^{\top}$.
It is obvious from \eqref{fnmt} that
\begin{align}\label{fbc} 
\sum_{\ell=0}^m a_\ell\bvec{f}_{m+1}^{(s+m-\ell,t)}=0.
\end{align}
Thus, combining \eqref{taum+1st} with \eqref{fbc} gives \eqref{taubc}.
\end{proof}
\par
The following proposition gives evolutions with respect to $s$ and $t$ in the Hankel determinants $H_k^{(s,t)}$ associated 
with the restricted infinite sequence $\{f_s^{(t)}\}_{s=0}^{\infty}$.
%
%
\begin{proposition}\label{prop2}
For $k=0,1,\dots,m$, it holds that
\begin{align}
&H_k^{(s+2,t)}H_k^{(s,t)}=H_{k}^{(s+1,t)}H_{k}^{(s+1,t)}+H_{k+1}^{(s,t)}H_{k-1}^{(s+2,t)},
\quad s,t=0,1,\dots,\label{tauJacobin}\\
&H_{k}^{(s,t+2)}H_{k}^{(s,t)}=H_{k}^{(s,t+1)}H_{k}^{(s,t+1)}+H_{k+1}^{(s,t)}H_{k-1}^{(s,t+2)},
\quad s,t=0,1,\dots.\label{tauJacobit}
\end{align}
\end{proposition}
\begin{proof}
For any determinant $D$, it is shown in \cite{Hirota2003} that
\begin{align}\label{Jacobi}
D\left[\begin{array}{c} i_1\\j_1\end{array}\right]
D\left[\begin{array}{c} i_2\\j_2\end{array}\right]
=
D\left[\begin{array}{c} i_1\\j_2\end{array}\right]
D\left[\begin{array}{c} i_2\\j_1\end{array}\right]
+
DD\left[\begin{array}{cc}i_1&i_2\\j_1&j_2\end{array}\right],
\end{align}
where $D\left[\begin{array}{cccc}i_1&i_2&\dots&i_n\\j_1&j_2&\dots&j_n\end{array}\right]$
denotes the cofactor obtained from $D$ by deleting the 
$i_1$th, $i_2$th, $\dots$, $i_n$th rows and $j_1$th, $j_2$th, $\dots$, $j_n$th columns.
Equation \eqref{Jacobi} is well known as the Jacobi identity for determinants.
Then, by letting $i_1=j_1=1$, $i_2=j_2=k+1$ and $D=H_{k+1}^{(s,t)}$ in \eqref{Jacobi},
we obtain
\begin{align}\label{temp}
H_{k+1}^{(s,t)}\left[\begin{array}{c} 1\\1\end{array}\right]
H_{k+1}^{(s,t)}\left[\begin{array}{c} k+1\\ k+1\end{array}\right]
=
H_{k+1}^{(s,t)}\left[\begin{array}{c} 1\\ k+1\end{array}\right]
H_{k+1}^{(s,t)}\left[\begin{array}{c} k+1\\ 1\end{array}\right]
+
H_{k+1}^{(s,t)}H_{k+1}^{(n,t)}\left[\begin{array}{cc}1&k+1\\1&k+1\end{array}\right].
\end{align}
It follows that
\begin{align*}
&\left|
\begin{array}{cccc}
\bvec{f}_k^{(s+2,t)}&\bvec{f}_{k}^{(s+3,t)}&\dots&\bvec{f}_{k}^{(s+k+1,t)}
\end{array}
\right|
\left|
\begin{array}{cccc}
\bvec{f}_k^{(s,t)}&\bvec{f}_{k}^{(s+1,t)}&\dots&\bvec{f}_{k}^{(s+k-1,t)}
\end{array}
\right|\\
&\quad=
\left|
\begin{array}{cccc}
\bvec{f}_k^{(s+1,t)}&\bvec{f}_{k}^{(s+2,t)}&\dots&\bvec{f}_{k}^{(s+k,t)}
\end{array}
\right|\left|
\begin{array}{cccc}
\bvec{f}_k^{(s+1,t)}&\bvec{f}_{k}^{(s+2,t)}&\dots&\bvec{f}_{k}^{(s+k,t)}
\end{array}
\right|\\
&\qquad
+\left|
\begin{array}{cccc}
\bvec{f}_{k+1}^{(s,t)}&\bvec{f}_{k+1}^{(s+1,t)}&\dots&\bvec{f}_{k+1}^{(s+k-1,t)}
\end{array}
\right|\left|
\begin{array}{cccc}
\bvec{f}_{k-1}^{(s+2,t)}&\bvec{f}_{k-1}^{(s+3,t)}&\dots&\bvec{f}_{k-1}^{(s+k,t)}
\end{array}
\right|,
\end{align*}
which immediately leads to \eqref{tauJacobin}.\par
By multiplying the $k$th, $(k-1)$th, \dots, $1$st rows by $-\mu^{(t)}$, and then adding these to the $(k+1)$th, $k$th, \dots, $2$nd rows in the entries of the Hankel determinants $H_{k+1}^{(s,t)}$ in \eqref{Hankel}, respectively, and then by considering \eqref{key},
we obtain 
\begin{align}\label{temp16}
H_{k+1}^{(s,t)}=\left|
\begin{array}{cccc}
f_{s}^{(t)}&f_{s+1}^{(t)}&\cdots&f_{s+k}^{(t)}\\[2pt]
f_{s}^{(t+1)}&f_{s+1}^{(t+1)}&\cdots&f_{s+k}^{(t+1)}\\
\vdots&\vdots&\vdots&\vdots\\
f_{s+k-1}^{(t+1)}&f_{s+k}^{(t+1)}&\cdots&f_{s+2k-1}^{(t+1)}
\end{array}
\right|.
\end{align}
Performing a similar procedure the columns in the right-hand side of \eqref{temp16},
we can express the Hankel determinants $H_{k+1}^{(s,t)}$ as
\begin{align}\label{temp1}
H_{k+1}^{(s,t)}=\left|
\begin{array}{ccccc}
f_{s}^{(t)}&f_{s}^{(t+1)}&f_{s+1}^{(t+1)}&\cdots&f_{s+k-1}^{(t+1)}\\[2pt]
f_{s}^{(t+1)}&f_{s}^{(t+2)}&f_{s+1}^{(t+2)}&\cdots&f_{s+k-1}^{(t+2)}\\[2pt]
f_{s+1}^{(t+1)}&f_{s+1}^{(t+2)}&f_{s+2}^{(t+2)}&\cdots&f_{s+k}^{(t+2)}\\
\vdots&\vdots&\vdots&\vdots&\vdots\\
f_{s+k-1}^{(t+1)}&f_{s+k-1}^{(t+2)}&f_{s+k}^{(t+2)}&\cdots&f_{s+2k-2}^{(t+2)}
\end{array}
\right|.
\end{align}
It is obvious from \eqref{temp1} that 
$H_{k+1}^{(s,t)}\left[\begin{array}{c}1\\1\end{array}\right]=H_{k}^{(s,t+2)}$ for $k=1,2,\dots,m$.
By considering the other cofactors obtained from $H_{k+1}^{(s,t)}$ in \eqref{temp1}, 
we can see that
\begin{align}
&H_{k+1}^{(s,t)}\left[\begin{array}{c} k+1\\ k+1\end{array}\right]=
\left|\begin{array}{ccccc}
f_{s}^{(t)}&f_{s}^{(t+1)}&f_{s+1}^{(t+1)}&\cdots&f_{s+k-2}^{(t+1)}\\[2pt]
\bvec{f}_{k-1}^{(s,t+1)}&\bvec{f}_{k-1}^{(s,t+2)}&\bvec{f}_{k-1}^{(s+1,t+2)}&\cdots&\bvec{f}_{k-1}^{(s+k-2,t+2)}
\end{array}
\right|,\label{temp2-17}\\
&H_{k+1}^{(s,t)}\left[\begin{array}{c} 1\\ k+1\end{array}\right]=
\left|\begin{array}{ccccc}
\bvec{f}_{k}^{(s,t+1)}&\bvec{f}_{k}^{(s,t+2)}&\bvec{f}_{k}^{(s+1,t+2)}&\cdots&\bvec{f}_{k}^{(s+k-2,t+2)}
\end{array}
\right|,\\
&H_{k+1}^{(s,t)}\left[\begin{array}{c} k+1\\ 1\end{array}\right]=
\left|\begin{array}{cccc}
f_{s}^{(t+1)}&f_{s+1}^{(t+1)}&\cdots&f_{s+k-1}^{(t+1)}\\[2pt]
\bvec{f}_{k-1}^{(s,t+2)}&\bvec{f}_{k-1}^{(s+1,t+2)}&\cdots&\bvec{f}_{k-1}^{(s+k-1,t+2)}
\end{array}
\right|,\\
&H_{k+1}^{(n,t)}\left[\begin{array}{cc}1&k+1\\1&k+1\end{array}\right]
=
\left|\begin{array}{cccc}
\bvec{f}_{k-1}^{(s,t+2)}&\bvec{f}_{k-1}^{(s+1,t+2)}&\cdots&\bvec{f}_{k-1}^{(s+k-2,t+2)}
\end{array}
\right|.
\end{align} 
By multiplying the $1$st row in the right-hand side of \eqref{temp2-17} by $\mu^{(t)}$, then adding it to the $2$nd row and by using \eqref{key},
we obtain
\begin{align}
H_{k+1}^{(s,t)}\left[\begin{array}{c} k+1\\ k+1\end{array}\right]
=\left|
\begin{array}{ccccc}
f_{s}^{(t)}&f_{s}^{(t+1)}&f_{s+1}^{(t+1)}&\cdots&f_{s+k-2}^{(t+1)}\\[2pt]
f_{s+1}^{(t)}&f_{s+1}^{(t+1)}&f_{s+2}^{(t+1)}&\cdots&f_{s+k-1}^{(t+1)}\\[2pt]
\bvec{f}_{k-2}^{(s+1,t+1)}&\bvec{f}_{k-2}^{(s+1,t+2)}&\bvec{f}_{k-2}^{(s+2,t+2)}&\cdots&\bvec{f}_{k-2}^{(s+k-1,t+2)}
\end{array}
\right|.\label{temp2-21}
\end{align}
Similarly, for the $2$nd, $3$rd, $\dots$, $k$th rows in \eqref{temp2-21},
we derive
\begin{align}
H_{k+1}^{(s,t)}\left[\begin{array}{c} k+1\\ k+1\end{array}\right]
=\left|
\begin{array}{ccccc}
f_{s}^{(t)}&f_{s}^{(t+1)}&f_{s+1}^{(t+1)}&\cdots&f_{s+k-2}^{(t+1)}\\[2pt]
f_{s+1}^{(t)}&f_{s+1}^{(t+1)}&f_{s+2}^{(t+1)}&\cdots&f_{s+k-1}^{(t+1)}\\[2pt]
\vdots&\vdots&\vdots&\ddots&\vdots\\
f_{s+k-1}^{(t)}&f_{s+k-1}^{(t+1)}&f_{s+k}^{(t+1)}&\cdots&f_{s+2k-3}^{(t+1)}
\end{array}
\right|.\label{temp2-22}
\end{align}
Performing a similar procedure for the $1$st, $2$nd, $\dots$, $(k-1)$th columns in \eqref{temp2-22},
it follows that 
$H_{k+1}^{(s,t)}\left[\begin{array}{c} k+1\\ k+1\end{array}\right]=H_{k}^{(s,t)}$.
Moreover, by rewriting the other cofactors as
$H_{k+1}^{(s,t)}\left[\begin{array}{c} 1\\ k+1\end{array}\right]=
H_{k+1}^{(s,t)}\left[\begin{array}{c} k+1\\ 1\end{array}\right]=H_{k}^{(s,t+1)}$
and 
$H_{k+1}^{(n,t)}\left[\begin{array}{cc}1&k+1\\1&k+1\end{array}\right]=H_{k-1}^{(s,t+2)}$,
we obtain \eqref{tauJacobit}.
\end{proof}

\section{Hadamard polynomials and autonomous discrete Toda}\label{sec:3}
In this section, we consider the Hankel polynomials associated with the Hankel determinants $H_k^{(s,t)}$ 
and define the Hadamard polynomials by rates of the Hankel determinants $H_k^{(s,t)}$ 
and their corresponding Hankel polynomials.
We then present relationships between the Hankel determinants $H_k^{(s,t)}$, 
the Hankel polynomials, and the Hadamard polynomials with respect to $s$. 
From these relationships, we derive the autonomous dToda equation \eqref{adToda}.
\par 
For $k=1,2,\dots,m$, as the Hankel polynomials associated with the Hankel determinants $H_k^{(s,t)}$,
let us introduce $k$-degree polynomials with respect to $z$ associated with the infinite sequence $\{f_{s}^{(t)}\}_{s,t=0}^{\infty}$:
\begin{align}
H_{k}^{(s,t)}(z):=\left|
\begin{array}{ccccc}
f_{s}^{(t)}&f_{s+1}^{(t)}&\cdots&f_{s+k-1}^{(t)}&1\\
f_{s+1}^{(t)}&f_{s+2}^{(t)}&\cdots&f_{s+k}^{(t)}&z\\
\vdots&\vdots&\vdots&\vdots&\vdots\\
f_{s+k}^{(t)}&f_{s+k+1}^{(t)}&\cdots&f_{s+2k-1}^{(t)}&z^k\
\end{array}
\right|,\quad s,t=0,1,\dots,\label{Tkz}
\end{align}
where $H_{-1}^{(s,t)}(z):=0$ and $H_{0}^{(s,t)}(z):=1$.\par
The following lemma gives the relationships of the Hankel determinants $H_k^{(s,t)}$ and the Hankel polynomials $H_k^{(s,t)}(z)$ at each $t$.
%
%
\begin{lemma}\label{lem1}
For $k=0,1,\dots,m$ and $s,t=0,1,\dots$, the Hankel determinants $H_k^{(s,t)}$ and the Hankel polynomials $H_{k}^{(s,t)}(z)$ 
associated with the infinite sequence $\{f_{s}^{(t)}\}_{s,t=0}^{\infty}$ satisfy
\begin{align}
&zH_k^{(s,t)}H_{k-1}^{(s+1,t)}(z)
=H_k^{(s+1,t)}H_{k-1}^{(s,t)}(z)+H_{k-1}^{(s+1,t)}H_k^{(s,t)}(z),
\label{TkzJacobi}\\
&H_{k}^{(s+1,t)}H_{k}^{(s,t)}(z)
=H_{k+1}^{(s,t)}H_{k-1}^{(s+1,t)}(z)+H_{k}^{(s,t)}H_{k}^{(s+1,t)}(z).\label{TkzPlucker}
\end{align}
\end{lemma}
\begin{proof}
By letting $i_1=j_1=k+1$, $i_2=1$, $j_2=k$ and $D=H_k^{(s,t)}(z)$ in \eqref{Jacobi},
we obtain
\begin{align*}
&H_k^{(s,t)}(z)\left[\begin{array}{c}k+1\\k+1\end{array}\right]
H_k^{(s,t)}(z)\left[\begin{array}{c}1\\k\end{array}\right]\\
&\quad=
H_k^{(s,t)}(z)\left[\begin{array}{c}k+1\\k\end{array}\right]
H_k^{(s,t)}(z)\left[\begin{array}{c}1\\k+1\end{array}\right]
+
H_k^{(s,t)}(z)
H_k^{(s,t)}(z)\left[\begin{array}{cc}k+1&1\\k+1&k\end{array}\right].
\end{align*}
Thus, it follows that
\begin{align*}
&\left|\begin{array}{cccc}
\bvec{f}_k^{(s,t)}&\bvec{f}_{k}^{(s+1,t)}&\cdots&\bvec{f}_k^{(s+k-1,t)}
\end{array}\right|\left|\begin{array}{ccccc}
\bvec{f}_k^{(s+1,t)}&\bvec{f}_{k}^{(s+2,t)}&\cdots&\bvec{f}_k^{(s+k-1,t)}&z\bvec{z}_{k}
\end{array}\right|\\
&=
\left|\begin{array}{ccccc}
\bvec{f}_k^{(s,t)}&\bvec{f}_{k}^{(s+1,t)}&\cdots&\bvec{f}_k^{(s+k-2,t)}&\bvec{z}_{k}
\end{array}\right|\left|\begin{array}{cccc}
\bvec{f}_k^{(s+1,t)}&\bvec{f}_{k}^{(s+2,t)}&\cdots&\bvec{f}_k^{(s+k,t)}
\end{array}\right|\\
&\qquad+\left|\begin{array}{ccccc}
\bvec{f}_{k+1}^{(s,t)}&\bvec{f}_{k+1}^{(s+1,t)}&\cdots&\bvec{f}_{k+1}^{(s+k-1,t)}&\bvec{z}_{k+1}
\end{array}\right|\left|\begin{array}{cccc}
\bvec{f}_{k-1}^{(s+1,t)}&\bvec{f}_{k-1}^{(s+2,t)}&\cdots&\bvec{f}_{k-1}^{(s+k-1,t)}
\end{array}\right|,
\end{align*}
where $\bvec{z}_{k}:=(1,z,\dots,z^{k-1})^{\top}$.
This identity is easily checked to be \eqref{TkzJacobi}.\par
For any $(k+1)$-dimensional vectors $\bvec{x}_1,\bvec{x}_2,\dots,\bvec{x}_{k+3}$,
the determinants of the submatrices in the $(k+1)$-by-$(k+3)$ matrix 
$X_k:=\left(\begin{array}{cccc}\bvec{x}_1&\bvec{x}_2&\cdots&\bvec{x}_{k+3}\end{array}\right)$ satisfy
\begin{align}\label{Plucker}
&\left|\begin{array}{cccccc}
\bvec{x}_1&\bvec{x}_2&\cdots&\bvec{x}_{k-1}&\bvec{x}_{k}&\bvec{x}_{k+1}
\end{array}\right|\left|\begin{array}{cccccc}
\bvec{x}_1&\bvec{x}_2&\cdots&\bvec{x}_{k-1}&\bvec{x}_{k+2}&\bvec{x}_{k+3}
\end{array}\right|\nonumber\\
&-\left|\begin{array}{cccccc}
\bvec{x}_1&\bvec{x}_2&\cdots&\bvec{x}_{k-1}&\bvec{x}_{k}&\bvec{x}_{k+2}
\end{array}\right|\left|\begin{array}{cccccc}
\bvec{x}_1&\bvec{x}_2&\cdots&\bvec{x}_{k-1}&\bvec{x}_{k+1}&\bvec{x}_{k+3}
\end{array}\right|\nonumber\\
&+\left|\begin{array}{cccccc}
\bvec{x}_1&\bvec{x}_2&\cdots&\bvec{x}_{k-1}&\bvec{x}_{k}&\bvec{x}_{k+3}
\end{array}\right|\left|\begin{array}{cccccc}
\bvec{x}_1&\bvec{x}_2&\cdots&\bvec{x}_{k-1}&\bvec{x}_{k+1}&\bvec{x}_{k+2}
\end{array}\right|=0.
\end{align}
This is known as the Pl\"ucker relation.
By letting $\bvec{x}_1=\bvec{f}_{k+1}^{(s+1,t)}$, $\bvec{x}_2=\bvec{f}_{k+1}^{(s+2,t)}$,
$\dots$, $\bvec{x}_{k}=\bvec{f}_{k+1}^{(s+k,t)}$, $\bvec{x}_{k+1}=\bvec{f}_{k+1}^{(s,t)}$,
$\bvec{x}_{k+2}=\bvec{z}_{k+1}$, $\bvec{x}_{k+3}=\bvec{e}_{k+1}:=(0,0,\dots,0,1)^{\top}$ in \eqref{Plucker},
we obtain 
\begin{align*}
&\left|\begin{array}{cccccc}
\bvec{f}_{k+1}^{(s+1,t)}&\bvec{f}_{k+1}^{(s+2,t)}&\cdots&\bvec{f}_{k+1}^{(s+k,t)}&\bvec{f}_{k+1}^{(s,t)}
\end{array}\right|\nonumber\\
&\quad\times\left|\begin{array}{cccccc}
\bvec{f}_{k+1}^{(s+1,t)}&\bvec{f}_{k+1}^{(s+2,t)}&\cdots&\bvec{f}_{k+1}^{(s+k-1,t)}&\bvec{z}_{k+1}&\bvec{e}_{k+1}
\end{array}\right|\nonumber\\
&-\left|\begin{array}{cccccc}
\bvec{f}_{k+1}^{(s+1,t)}&\bvec{f}_{k+1}^{(s+2,t)}&\cdots&\bvec{f}_{k+1}^{(s+k-1,t)}&\bvec{f}_{k+1}^{(s+k,t)}&\bvec{z}_{k+1}
\end{array}\right|\nonumber\\
&\quad\times\left|
\begin{array}{cccccc}
\bvec{f}_{k+1}^{(s+1,t)}&\bvec{f}_{k+1}^{(s+2,t)}&\cdots&\bvec{f}_{k+1}^{(s+k-1,t)}&\bvec{f}_{k+1}^{(s,t)}&\bvec{e}_{k+1}
\end{array}\right|\nonumber\\
&+\left|\begin{array}{cccccc}
\bvec{f}_{k+1}^{(s+1,t)}&\bvec{f}_{k+1}^{(s+2,t)}&\cdots&\bvec{f}_{k+1}^{(s+k-1,t)}&\bvec{f}_{k+1}^{(s+k,t)}&\bvec{e}_{k+1}
\end{array}\right|\nonumber\\
&\quad\times\left|
\begin{array}{cccccc}
\bvec{f}_{k+1}^{(s+1,t)}&\bvec{f}_{k+1}^{(s+2,t)}&\cdots&\bvec{f}_{k+1}^{(s+k-1,t)}&\bvec{f}_{k+1}^{(s,t)}&\bvec{z}_{k+1}
\end{array}\right|=0,
\end{align*}
which implies \eqref{TkzPlucker}.

\end{proof}
\par
Moreover, for $k=1,2,\dots,m$, let us define another set of polynomials ${\cal H}_{k}^{(s,t)}(z)$ of degree $k$ as 
\begin{align}
{\cal H}_{k}^{(s,t)}(z):=\frac{H_{k}^{(s,t)}(z)}{H_{k}^{(s,t)}},\quad s,t=0,1,\dots,\label{HadamardHankel}
\end{align}
where ${\cal H}_{-1}^{(s,t)}(z):= 0$ and ${\cal H}_{0}^{(s,t)}(z):=1$.
The polynomials ${\cal H}_k^{(s,t)}(z)$ are the Hadamard polynomials.
Now, we derive a proposition for the Hadamard polynomials ${\cal H}_k^{(s,t)}(z)$ with $k=m$.
%
%
\begin{proposition}\label{prop3}
The Hadamard polynomial ${\cal H}_m^{(s,t)}(z)$ coincides with
the characteristic polynomial $p(z)$ of matrices with eigenvalues $\lambda_1,\lambda_2,\dots,\lambda_m,$
\begin{align}
{\cal H}_{m}^{(s,t)}(z)=p(z).\label{phizpz}
\end{align}
\end{proposition}
\begin{proof}
The $m$-degree polynomial $H_m^{(s,t)}(z)$ can be written as 
\begin{align}\label{tmp6}
H_m^{(s,t)}(z)=
\left|\begin{array}{ccccc}
\bvec{f}_{m+1}^{(s,t)}&\bvec{f}_{m+1}^{(s+1,t)}&\cdots&\bvec{f}_{m+1}^{(s+m-1,t)}&\bvec{z}_{m+1}
\end{array}
\right|.\end{align}
By multiplying the $1$st, $2$nd, $\dots$, $m$th rows on the right-hand side of \eqref{tmp6} by $a_m$, $a_{m-1}$, $\dots$, $a_1$, respectively,
and adding these to the $(m+1)$th row,
we obtain
\begin{align*}
H_m^{(s,t)}(z)
=\left|\begin{array}{ccccc}
\bvec{f}_{m}^{(s,t)}&\bvec{f}_{m}^{(s+1,t)}&\cdots&\bvec{f}_{m}^{(s+m-1,t)}&\bvec{z}_{m}\\
\displaystyle \sum_{\ell=0}^{m}a_{\ell}f_{s+m-\ell}^{(t)}
&\displaystyle \sum_{\ell=0}^{m}a_{\ell}f_{s+m+1-\ell}^{(t)}
&\cdots
&\displaystyle \sum_{\ell=0}^{m}a_{\ell}f_{s+2m-1-\ell}^{(t)}
&\displaystyle \sum_{\ell=0}^{m}a_{\ell}z^{m-\ell}
\end{array}\right|.
\end{align*}
From \eqref{fnmt}, it is obvious that the ($m+1,1$), ($m+1,2$), $\dots$, ($m+1,m$) entries are $0$. 
The ($m+1,m+1$) entry is equal to the characteristic polynomial $p(z)$. 
Thus, we can express the characteristic polynomial $p(z)$ 
using the Hankel determinant $H_{m}^{(s,t)}$ and the Hankel polynomial $H_{m}^{(s,t)}(z)$ as $H_m^{(s,t)}(z)=p(z)H_m^{(s,t)}$.
By recalling the definition of the Hadamard polynomial ${\cal H}_m^{(s,t)}(z)$ in \eqref{HadamardHankel}, 
we obtain \eqref{phizpz}.

\end{proof}
\par
The following lemma gives identities of the Hadamard polynomials ${\cal H}_k^{(s,t)}(z)$.
%
%
\begin{lemma}\label{lem3}
Let us assume that for the infinite sequence $\{f_s^{(t)}\}_{s,t=0}^{\infty}$, 
the Hankel determinants $H_k^{(s,t)}$ are all nonzero.
Then, for $k=1,2,\dots,m$,
the Hadamard polynomials ${\cal H}_k^{(s,t)}(z)$ satisfy
\begin{align}
&z{\cal H}_{k-1}^{(s+1,t)}(z)={\cal H}_{k}^{(s,t)}(z)+q_k^{(s,t)}{\cal H}_{k-1}^{(s,t)}(z),
\quad s,t=0,1,\dots,\label{Christoffeln}\\
&{\cal H}_k^{(s,t)}(z)={\cal H}_{k}^{(s+1,t)}(z)+e_{k}^{(s,t)}{\cal H}_{k-1}^{(s+1,t)}(z),
\quad s,t=0,1,\dots,
\label{Geronimusn}
\end{align}
where $q_k^{(s,t)}$ and $e_k^{(s,t)}$ are given using the Hankel determinants $H_k^{(s,t)}$ as
\begin{align}
&q_{k}^{(s,t)}:=\frac{H_{k-1}^{(s,t)}H_{k}^{(s+1,t)}}{H_{k}^{(s,t)}H_{k-1}^{(s+1,t)}},\label{qkst}\\
&e_k^{(s,t)}:=\frac{H_{k+1}^{(s,t)}H_{k-1}^{(s+1,t)}}{H_{k}^{(s,t)}H_{k}^{(s+1,t)}}.\label{ekst}
\end{align}
\end{lemma}
\begin{proof}
By dividing both sides of \eqref{TkzJacobi} by $H_{k-1}^{(s+1,t)}H_{k}^{(s,t)}$,
we obtain 
\begin{align*}
z\frac{H_{k-1}^{(s+1,t)}(z)}{H_{k-1}^{(s+1,t)}}=
\left(\frac{H_{k-1}^{(s,t)}H_{k}^{(s+1,t)}}{H_{k}^{(s,t)}H_{k-1}^{(s+1,t)}}\right)
\frac{H_{k-1}^{(s,t)}(z)}{H_{k-1}^{(s,t)}}+\frac{H_{k}^{(s,t)}(z)}{H_{k}^{(s,t)}}.
\end{align*}
By considering the definition of the Hadamard polynomials ${\cal H}_k^{(s,t)}(z)$ in \eqref{HadamardHankel},
we thus derive \eqref{Christoffeln} with \eqref{qkst}. 
Similarly, it follows from \eqref{TkzPlucker} that 
\begin{align}\label{tmp3}
\frac{H_{k+1}^{(s,t)}H_{k-1}^{(s+1,t)}(z)}{H_{k}^{(s,t)}H_{k}^{(s+1,t)}}
={\cal H}_k^{(s,t)}(z)-{\cal H}_{k}^{(s+1,t)}(z).
\end{align}
The left-hand side of \eqref{tmp3} can be rewritten as
\begin{align*}
\frac{H_{k+1}^{(s,t)}H_{k-1}^{(s+1,t)}(z)}{H_{k}^{(s,t)}H_{k}^{(s+1,t)}}
=\left(\frac{H_{k+1}^{(s,t)}H_{k-1}^{(s+1,t)}}{H_{k}^{(s,t)}H_{k}^{(s+1,t)}}\right)
{\cal H}_{k-1}^{(s+1,t)},\end{align*}
and so we have \eqref{Geronimusn} with \eqref{ekst}.
 
\end{proof}

The identities \eqref{Christoffeln} in Lemma~\ref{lem3} can be regarded as the recursion formulas for generating evolutions of discrete-time variable $s$ of the Hadamard polynomials ${\cal H}_k^{(s,t)}(z)$.
In the theory of orthogonal polynomials, \eqref{Christoffeln} correspond to 
the Christoffel transformations for the Hadamard polynomials ${\cal H}_k^{(s,t)}(z)$ in \cite{Chihara}.
Similarly, \eqref{Geronimusn}, which is the inverse transformation of \eqref{Christoffeln},
is the Geronimus transformation for the Hadamard polynomials ${\cal H}_k^{(s,t)}(z)$.
The Christoffel and Geronimus transformations are useful in the study of integrable discrete systems \cite{Zhedanov}.
\par
From Lemma~\ref{lem3}, we derive a theorem for relationships involving $q_k^{(s,t)}$ and $e_k^{(s,t)}$.
%
%
\begin{theorem}\label{lem5}
Let us assume that for the infinite sequence $\{f_{s}^{(t)}\}_{s,t=0}^{\infty}$, 
the Hankel determinants $H_k^{(s,t)}$ are all nonzero.
Then, the Hadamard polynomials ${\cal H}_k^{(s,t)}(z)$ for $s,t=0,1,\dots$ satisfy the three-term recursion formulas:
\begin{align}\label{3termqe}
{\cal H}_{k+1}^{(s,t)}(z)=(z-q_{k+1}^{(s,t)}-e_{k}^{(s,t)}){\cal H}_{k}^{(s,t)}(z)
-q_{k}^{(s,t)}e_{k}^{(s,t)}{\cal H}_{k-1}^{(s,t)}(z),
\quad k=0,1,\dots,m-1.
\end{align}
More importantly, $q_k^{(s,t)}$ and $e_k^{(s,t)}$ satisfy
\begin{align}\label{q+en}
\left\{
\begin{array}{l}
q_k^{(s+1,t)}+e_{k-1}^{(s+1,t)}=q_k^{(s,t)}+e_k^{(s,t)},\quad k=1,2,\dots,m,\\[2pt]
q_k^{(s+1,t)}e_{k}^{(s+1,t)}=q_{k+1}^{(s,t)}e_{k}^{(s,t)},\quad k=1,2,\dots,m-1,\\[2pt]
s,t=0,1,\dots.
\end{array}
\right.
\end{align}
\end{theorem}
\begin{proof}
Substituting \eqref{Christoffeln} for \eqref{Geronimusn},
we obtain
\begin{align}\label{CG1}
z{\cal H}_k^{(s,t)}(z)={\cal H}_{k+1}^{(s,t)}(z)+(q_{k+1}^{(s,t)}+e_{k}^{(s,t)}){\cal H}_{k}^{(s,t)}(z)+q_{k}^{(s,t)}e_k^{(s,t)}{\cal H}_{k-1}^{(s,t)}(z).
\end{align}
By replacing $s$ and $k$ with $s+1$ and $k-1$ in \eqref{CG1}, respectively,
we derive
\begin{align}\label{tmp7}
z{\cal H}_{k-1}^{(s+1,t)}(z)
={\cal H}_{k}^{(s+1,t)}(z)+(q_{k}^{(s+1,t)}+e_{k-1}^{(s+1,t)}){\cal H}_{k-1}^{(s+1,t)}(z)
+q_{k-1}^{(s+1,t)}e_{k-1}^{(s+1,t)}{\cal H}_{k-2}^{(s+1,t)}(z).
\end{align}
Moreover, by substituting \eqref{Geronimusn} for \eqref{Christoffeln},
we obtain
\begin{align}\label{CG2}
z{\cal H}_{k-1}^{(s+1,t)}(z)
={\cal H}_{k}^{(s+1,t)}(z)+(q_{k}^{(s,t)}+e_{k}^{(s,t)}){\cal H}_{k-1}^{(s+1,t)}(z)+q_{k}^{(s,t)}e_{k-1}^{(s,t)}{\cal H}_{k-2}^{(s+1,t)}(z).
\end{align}
By observing both sides of \eqref{tmp7} and \eqref{CG2},
we have \eqref{q+en}.
\end{proof}
\par
Here, since \eqref{q+en} does not give evolution with respect to $t$,
we may simplify $q_k^{(s,t)}$ and $e_k^{(s,t)}$ as $q_k^{(s)}$ and $e_k^{(s)}$, respectively, 
in the case where $t$ is fixed in Theorem~\ref{lem5}.
Thus, \eqref{q+en} with fixed $t$ in Theorem~\ref{lem5} is equivalent to the autonomous dToda equation \eqref{adToda}.
In other words, the autonomous dToda equation \eqref{adToda} is given by a special case of \eqref{q+en}.

\section{Hadamard polynomials and non-autonomous discrete Toda}\label{sec:31}
In this section, with respect to $t$, but not $s$,
we present relationships between the Hankel determinants $H_k^{(s,t)}$, 
the Hankel polynomials  $H_k^{(s,t)}(z)$, and the Hadamard polynomials ${\cal H}_k^{(s,t)}(z)$ defined in Section~\ref{sec:3}. 
From these relationships, we also derive the non-autonomous dToda equation \eqref{nadToda}.\par
We derive a lemma for the Hankel determinants $H_k^{(s,t)}$ and the Hankel polynomials $H_{k}^{(s,t)}(z)$ at each $s$.
%
%
\begin{lemma}\label{lem2}
For $k=1,2,\dots,m$ and $s,t=0,1,\dots$, the Hankel polynomials $H_{k}^{(s,t)}(z)$ associated with the infinite sequence $\{f_{s}^{(t)}\}_{s,t=0}^{\infty}$ satisfy
\begin{align}
&(z-\mu^{(t)})H_{k}^{(s,t)}H_{k-1}^{(s,t+1)}(z)=H_{k}^{(s,t+1)}H_{k-1}^{(s,t)}(z)+H_{k-1}^{(s,t+1)}H_{k}^{(s,t)}(z),\label{TkzJacobit}\\
&H_{k}^{(s,t+1)}H_k^{(s,t)}(z)=H_{k+1}^{(s,t)}H_{k-1}^{(s,t+1)}(z)+H_{k}^{(s,t)}H_{k}^{(s,t+1)}(z).\label{TkzPluckert}
\end{align}
\end{lemma}
\begin{proof}
By multiplying the $k$th row in the Hankel polynomials $H_k^{(s,t)}(z)$ by $-\mu^{(t)}$, 
then adding it to the $(k+1)$th row and by using \eqref{key},
we derive
\begin{align*}
H_k^{(s,t)}(z)=
\left|\begin{array}{ccccc}
f_s^{(t)}&f_{s+1}^{(t)}&\dots&f_{s+k-1}^{(t)}&1\\
f_{s+1}^{(t)}&f_{s+2}^{(t)}&\dots&f_{s+k}^{(t)}&z\\
\vdots&\vdots&\ddots&\vdots&\vdots\\
f_{s+k-1}^{(t)}&f_{s+k}^{(t)}&\dots&f_{s+2k-2}^{(t)}&z^{k-1}\\
f_{s+k-1}^{(t+1)}&f_{s+k}^{(t+1)}
&\dots&f_{s+2k-2}^{(t+1)}&(z-\mu^{(t)})z^{k-1}
\end{array}\right|.
\end{align*}
Similarly, for the $k$th, $(k-1)$th, $(k-2)$th, $\dots$, $2$nd rows in the Hankel polynomials $H_k^{(s,t)}(z)$,
it follows that
\begin{align}
H_{k}^{(s,t)}(z)
=\left|\begin{array}{ccccc}
f_{s}^{(t)}&f_{s+1}^{(t)}&\cdots&f_{s+k-1}^{(t)}&1\\
\bvec{f}_{k}^{(s,t+1)}&\bvec{f}_{k}^{(s+1,t+1)}&\cdots&\bvec{f}_{k}^{(s+k-1,t+1)}&
\left(z-\mu^{(t)}\right)\bvec{z}_k
\end{array}\right|.\label{Tkzjacobitproof}
\end{align}
By letting $i_1=j_1=k+1$, $i_2=1$, $j_2=k$ and $D=H_{k}^{(s,t)}(z)$ 
in the Jacobi identity \eqref{Jacobi} and considering \eqref{Tkzjacobitproof},
we obtain
\begin{align}
&\left|\begin{array}{ccccc}
f_s^{(t)} & f_{s+1}^{(t)}&\cdots&f_{s+k-1}^{(t)}\\
\bvec{f}_{k-1}^{(s,t+1)}&\bvec{f}_{k-1}^{(s+1,t+1)}&\cdots&\bvec{f}_{k-1}^{(s+k-1,t+1)}
\end{array}\right|\nonumber\\
&\quad\times\left|
\begin{array}{ccccc}
\bvec{f}_{k}^{(s,t+1)}&\bvec{f}_{k}^{(s+1,t+1)}&\cdots&\bvec{f}_{k}^{(s+k-2,t+1)}&(z-\mu^{(t)})\bvec{z}_k
\end{array}\right|\nonumber\\
&=\left|\begin{array}{ccccc}
f_s^{(t)} & f_{s+1}^{(t)}&\cdots&f_{s+k-2}^{(t)}&1\\
\bvec{f}_{k-1}^{(s,t+1)}&\bvec{f}_{k-1}^{(s+1,t+1)}&\cdots&\bvec{f}_{k-1}^{(s+k-2,t+1)}&
(z-\mu^{(t)})\bvec{z}_{k-1}
\end{array}\right|\nonumber\\
&\quad\times\left|\begin{array}{ccccc}
\bvec{f}_{k}^{(s,t+1)}&\bvec{f}_{k}^{(s+1,t+1)}&\cdots&\bvec{f}_{k}^{(s+k-1,t+1)}
\end{array}\right|\nonumber\\
&+\left|\begin{array}{ccccc}
f_s^{(t)} & f_{s+1}^{(t)}&\cdots&f_{s+k-1}^{(t)}&1\\
\bvec{f}_{k-1}^{(s,t+1)}&\bvec{f}_{k-1}^{(s+1,t+1)}&\cdots&\bvec{f}_{k-1}^{(s+k-1,t+1)}
&(z-\mu^{(t)})\bvec{z}_k
\end{array}\right|\nonumber\\
&\quad\times\left|\begin{array}{ccccc}
\bvec{f}_{k-1}^{(s,t+1)}&\bvec{f}_{k-1}^{(s+1,t+1)}&\cdots&\bvec{f}_{k-1}^{(s+k-2,t+1)}
\end{array}\right|.\label{tmp4}
\end{align} 
Equation \eqref{tmp4} immediately leads to \eqref{TkzJacobit}.\par
By multiplying the $k$th, $(k-1)$th, $\dots$, $1$st columns in the Hankel polynomials $H_k^{(s,t)}(z)$ by $-\mu^{(t)}$
and adding these to the $(k+1)$th, $k$th, $\dots$, $2$nd columns, respectively,
we derive
\begin{align}\label{tmp5}
H_k^{(s,t)}(z)=\left|\begin{array}{cccccc}
\bvec{f}_{k+1}^{(s,t)}&\bvec{f}_{k+1}^{(s,t+1)}&\bvec{f}_{k+1}^{(s+1,t+1)}&\cdots&
\bvec{f}_{k+1}^{(s+k-2,t+1)}&\bvec{z}_{k+1}
\end{array}\right|.
\end{align}
By letting 
$\bvec{x}_1=\bvec{f}_{k+1}^{(s,t+1)}$, $\bvec{x}_2=\bvec{f}_{k+1}^{(s+1,t+1)}$,
$\dots$, $\bvec{x}_k=\bvec{f}_{k+1}^{(s+k-1,t+1)}$, $\bvec{x}_{k+1}=\bvec{f}_{k+1}^{(s,t)}$,
$\bvec{x}_{k+2}=\bvec{z}_{k+1}$, $\bvec{x}_{k+3}=\bvec{e}_{k+1}$
in the Pl\"uker relation \eqref{Plucker}, we obtain
\begin{align}
&\left|\begin{array}{cccccc}
\bvec{f}_{k+1}^{(s,t+1)}&\bvec{f}_{k+1}^{(s+1,t+1)}&
\cdots&\bvec{f}_{k+1}^{(s+k-1,t+1)}&\bvec{f}_{k+1}^{(s,t)}
\end{array}\right|\nonumber\\
&\quad \times \left|\begin{array}{cccccc}
\bvec{f}_{k+1}^{(s,t+1)}&\bvec{f}_{k+1}^{(s+1,t+1)}&
\cdots&\bvec{f}_{k+1}^{(s+k-2,t+1)}&\bvec{z}_{k+1}&
\bvec{e}_{k+1}
\end{array}\right|\nonumber\\
&-\left|\begin{array}{cccccccc}
\bvec{f}_{k+1}^{(s,t+1)}&\bvec{f}_{k+1}^{(s+1,t+1)}&
\cdots&\bvec{f}_{k+1}^{(s+k-1,t+1)}&
\bvec{z}_{k+1}
\end{array}\right|\nonumber\\
&\quad \times\left|\begin{array}{cccccccc}
\bvec{f}_{k+1}^{(s,t+1)}&\bvec{f}_{k+1}^{(s+1,t+1)}&
\cdots&\bvec{f}_{k+1}^{(s+k-2,t+1)}&
\bvec{f}_{k+1}^{(s,t)}&\bvec{e}_{k+1}
\end{array}\right|\nonumber\\
&+\left|\begin{array}{cccccccc}
\bvec{f}_{k+1}^{(s,t+1)}&\bvec{f}_{k+1}^{(s+1,t+1)}&
\cdots&\bvec{f}_{k+1}^{(s+k-1,t+1)}&\bvec{e}_{k+1}
\end{array}\right|\nonumber\\
&\quad \times\left|\begin{array}{cccccccc}
\bvec{f}_{k+1}^{(s,t+1)}&\bvec{f}_{k+1}^{(s+1,t+1)}&
\cdots&\bvec{f}_{k+1}^{(s+k-2,t+1)}&
\bvec{f}_{k+1}^{(s,t)}&\bvec{z}_{k+1}
\end{array}\right|=0.\label{Tkzprooftem4}
\end{align}
Thus, combining \eqref{tmp5} with \eqref{Tkzprooftem4} gives \eqref{TkzPluckert}.

\end{proof}
\par
We give a lemma concerning identities of the Hadamard polynomials ${\cal H}_k^{(s,t)}$.
%
%
\begin{lemma}\label{lem4}
Let us assume that for the infinite sequence $\{f_{s}^{(t)}\}_{s,t=0}^{\infty}$,
the Hankel determinants $H_k^{(s,t)}$ are all nonzero.
Then, for $k=1,2,\dots,m$,
the Hadamard polynomials ${\cal H}_k^{(s,t)}(z)$ satisfy
\begin{align}
&(z-\mu^{(t)}){\cal H}_{k-1}^{(s,t+1)}(z)={\cal H}_{k}^{(s,t)}(z)+Q_k^{(s,t)}{\cal H}_{k-1}^{(s,t)}(z),
\quad s,t=0,1,\dots,\label{Christoffelt}\\
&{\cal H}_{k}^{(s,t)}(z)
={\cal H}_{k}^{(s,t+1)}(z)+E_{k}^{(s,t)}{\cal H}_{k-1}^{(s,t+1)}(z),
\quad s,t=0,1,\dots,\label{Geronimust}
\end{align}
where $Q_k^{(s,t)}$ and $E_k^{(s,t)}$ are given using the Hankel determinants $H_k^{(s,t)}$ as
\begin{align}
&Q_k^{(s,t)}:=\frac{H_{k-1}^{(s,t)}H_{k}^{(s,t+1)}}{H_{k}^{(s,t)}H_{k-1}^{(s,t+1)}},\label{Qkst}\\
&E_k^{(s,t)}:=
\frac{H_{k+1}^{(s,t)}H_{k-1}^{(s,t+1)}}{H_{k}^{(s,t)}H_{k}^{(s,t+1)}}.\label{Ekst}
\end{align}
\end{lemma}
\begin{proof}
By dividing both sides of \eqref{TkzJacobit} by $H_{k-1}^{(s,t+1)}H_{k}^{(s,t)}$
and considering the Hadamard polynomials ${\cal H}_k^{(s,t)}(z)$ in \eqref{HadamardHankel},
we obtain
\begin{align*}
(z-\mu^{(t)}){\cal H}_{k-1}^{(s,t+1)}(z)=
\left(\frac{H_{k-1}^{(s,t)}H_{k}^{(s,t+1)}}{H_{k}^{(s,t)}H_{k-1}^{(s,t+1)}}\right)
{\cal H}_{k-1}^{(s,t)}(z)+{\cal H}_{k}^{(s,t)}(z),
\end{align*}
which immediately leads to \eqref{Christoffelt} with \eqref{Qkst}.
Similarly, it follows from \eqref{TkzPluckert} that
\begin{align*}
\left(\frac{H_{k+1}^{(s,t)}H_{k-1}^{(s,t+1)}}{H_{k}^{(s,t)}H_{k}^{(s,t+1)}}\right)
{\cal H}_{k-1}^{(s,t+1)}(z)={\cal H}_{k}^{(s,t)}(z)-{\cal H}_{k}^{(s,t+1)}(z),
\end{align*}
which is equivalent to \eqref{Geronimust} with \eqref{Ekst}.

\end{proof}
\par
With the help of Lemma~\ref{lem4},
we give the following theorem concerning $Q_k^{(s,t)}$ and $E_k^{(s,t)}$.
\par
%
%
\begin{theorem}\label{lem6}
Let us assume that for the infinite sequence $\{f_{s}^{(t)}\}_{s,t=0}^{\infty}$, 
the Hankel determinants $H_k^{(s,t)}$ are all nonzero.
Then, the Hadamard polynomials ${\cal H}_k^{(s,t)}(z)$ for $s,t=0,1,\dots$ satisfy the three-term recursion formulas:
\begin{align}\label{3termQE}
{\cal H}_{k+1}^{(s,t)}(z)=(z-Q_{k+1}^{(s,t)}-E_{k}^{(s,t)}-\mu^{(t)}){\cal H}_{k}^{(s,t)}(z)
-Q_{k}^{(s,t)}E_{k}^{(s,t)}{\cal H}_{k-1}^{(s,t)}(z),
\quad k=0,1,\dots,m-1.
\end{align}
More importantly, $Q_k^{(s,t)}$ and $E_k^{(s,t)}$ satisfy 
\begin{align}\label{Q+Et}
\left\{
\begin{array}{l}
Q_{k}^{(s,t+1)}+E_{k-1}^{(s,t+1)}+\mu^{(t+1)}
=Q_{k}^{(s,t)}+E_{k}^{(s,t)}+\mu^{(t)},\quad k=1,2,\dots,m,\\[2pt]
Q_{k}^{(s,t+1)}E_{k}^{(s,t+1)}
=Q_{k+1}^{(s,t)}E_{k-1}^{(s,t)},\quad k=1,2,\dots,m-1,\\[2pt]
s,t=0,1,\dots.
\end{array}
\right.
\end{align}
\end{theorem}
\begin{proof}
Similar to the proof for Lemma~\ref{lem5},
it follows from \eqref{Christoffelt} and \eqref{Geronimust} that
\begin{align}
(z-\mu^{(t)}){\cal H}_{k}^{(s,t)}(z)
={\cal H}_{k+1}^{(s,t)}(z)+(Q_{k+1}^{(s,t)}+E_{k}^{(s,t)}){\cal H}_{k}^{(s,t)}(z)
+Q_{k}^{(s,t)}E_{k}^{(s,t)}{\cal H}_{k-1}^{(s,t)}(z),\label{lem6-1}\\
(z-\mu^{(t)}){\cal H}_{k-1}^{(s,t+1)}(z)
={\cal H}_{k}^{(s,t+1)}(z)+(Q_{k}^{(s,t)}+E_{k}^{(s,t)}){\cal H}_{k-1}^{(s,t+1)}(z)
+Q_{k}^{(s,t)}E_{k-1}^{(s,t)}{\cal H}_{k-2}^{(s,t)}(z).\label{gct}
\end{align}
Comparing \eqref{lem6-1} with \eqref{gct} gives \eqref{Q+Et}.

\end{proof}
\par
By considering $Q_k^{(s,t)}, E_k^{(s,t)}$ as $Q_k^{(t)}, E_k^{(t)}$ at each $s$, 
we immediately simplify \eqref{Q+Et} as the non-autonomous dToda equation \eqref{nadToda}.\par
Theorems~\ref{lem5} and~\ref{lem6} suggest that the two types of discrete-time variables play a key role in simultaneously deriving 
both the autonomous dToda equation \eqref{adToda} and the non-autonomous dToda equation \eqref{nadToda}. 

\section{Determinant solutions}\label{sec:4}
In this section, we clarify eigenpairs of tridiagonal matrices involving $q_k^{(s,t)},e_k^{(s,t)}$ in Theorem~\ref{lem5} and $Q_k^{(s,t)},E_k^{(s,t)}$ in Theorem~\ref{lem6}.
We then show the determinant solutions with sufficient degrees of freedom to the autonomous dToda equation \eqref{adToda}
and the non-autonomous dToda equation \eqref{nadToda}. \par
Let us introduce $m$-by-$m$ bidiagonal matrices with $e_1^{(s,t)}$, $e_2^{(s,t)},\dots,e_{m-1}^{(s,t)}$ and $q_1^{(s,t)},q_2^{(s,t)},\dots,q_m^{(s,t)}$:
\begin{align*}
L^{(s,t)}:=\left(
\begin{array}{cccc}
1&\\
e_1^{(s,t)}&1&\\
&\ddots&\ddots&\\
&&e_{m-1}^{(s,t)}&1
\end{array}
\right),\quad
R^{(s,t)}:=\left(
\begin{array}{cccc}
q_1^{(s,t)}&1\\
&q_2^{(s,t)}&\ddots\\
&&\ddots&1\\
&&&q_m^{(s,t)}
\end{array}
\right).
\end{align*}
Then, the following proposition shows eigenpairs of $A^{(s,t)}:=L^{(s,t)}R^{(s,t)}$.
%
%
\begin{proposition}\label{prop4}
Eigenvalues of $A^{(s,t)}$ coincide with the roots $\lambda_1,\lambda_2,\dots,\lambda_m$ of the polynomial
$p(z)=(z-\lambda_1)(z-\lambda_2)\dots(z-\lambda_m)$.
Moreover, the eigenvectors corresponding to $\lambda_1,\lambda_2,\dots,\lambda_m$ 
are ${\cal H}_1^{(s,t)},{\cal H}_2^{(s,t)},\dots,{\cal H}_m^{(s,t)}$, respectively, 
where ${\cal H}_k^{(s,t)}:=({\cal H}_0^{(s,t)}(\lambda_k),{\cal H}_1^{(s,t)}(\lambda_k),\dots,{\cal H}_{m-1}^{(s,t)}(\lambda_k))^\top$.
\end{proposition}
\begin{proof}
Let us prepare the $k$th principal submatrix of $A^{(s,t)}$, namely,
\begin{align*}
A_k^{(s,t)}
:=\left(
\begin{array}{cccc}
q_1^{(s,t)}&1\\
q_1^{(s,t)}e_{1}^{(s,t)}&q_2^{(s,t)}+e_1^{(s,t)}&\ddots\\
&\ddots&\ddots&1\\
&&q_{k-1}^{(s,t)}e_{k-1}^{(s,t)}&q_k^{(s,t)}+e_{k-1}^{(s,t)}
\end{array}
\right).
\end{align*}
Moreover, let $I_k$ be the $k$-by-$k$ identity matrix.
Then, by considering cofactor expansion along the $(k+1)$th row of $\det(zI_{k+1}-A_{k+1}^{(s,t)})$, we derive
\begin{align}\label{phizI-A}
\det(zI_{k+1}-A_{k+1}^{(s,t)})
=(z-q_{k+1}^{(s,t)}-e_k^{(s,t)})\det(zI_k-A_{k}^{(s,t)})
-q_k^{(s,t)}e_k^{(s,t)}\det(zI_{k-1}-A_{k-1}^{(s,t)}).
\end{align}
By comparing \eqref{3termqe} in Theorem~\ref{lem5} with \eqref{phizI-A},
we obtain
\begin{align}
{\cal H}_{k}^{(s,t)}(z)=\det(zI_k-A_k^{(s,t)}),\quad k=1,2,\dots,m.\label{phikstdet}
\end{align}
Thus, by combining Proposition~\ref{prop3} with \eqref{phikstdet},
we have
\begin{align}
p(z)=\det(zI_m-A^{(s,t)}).
\end{align}
Consequently, by taking into account that $p(z)=(z-\lambda_1)(z-\lambda_2)\cdots(z-\lambda_m)$, 
we see that
$\lambda_1,\lambda_2,\dots,\lambda_m$ are the eigenvalues of $A^{(s,t)}$.\par
Equation \eqref{3termqe} with $z=\lambda_k$ leads to
\begin{align}\label{qephisystem}
\left\{
\begin{array}{l}
(q_{1}^{(s,t)}+e_{0}^{(s,t)}){\cal H}_{0}^{(s,t)}(\lambda_k)+{\cal H}_{1}^{(s,t)}(\lambda_k)=\lambda_k{\cal H}_{0}^{(s,t)}(\lambda_k),\\
q_{1}^{(s,t)}e_{1}^{(s,t)}{\cal H}_{0}^{(s,t)}(\lambda_k)+
(q_{2}^{(s,t)}+e_{1}^{(s,t)}){\cal H}_{1}^{(s,t)}(\lambda_k)+{\cal H}_{2}^{(s,t)}(\lambda_k)=\lambda_k{\cal H}_{1}^{(s,t)}(\lambda_k),\\
\quad\vdots\\
q_{m-1}^{(s,t)}e_{m-1}^{(s,t)}{\cal H}_{m-2}^{(s,t)}(\lambda_k)+
(q_{m}^{(s,t)}+e_{m-1}^{(s,t)}){\cal H}_{m-1}^{(s,t)}(\lambda_k)=\lambda_k{\cal H}_{m-1}^{(s,t)}(\lambda_k).
\end{array}
\right.
\end{align}
Noting that $A^{(s,t)}=A_m^{(s,t)}$, we can rewrite \eqref{qephisystem} as 
\begin{align}\label{eigenpqe}
A^{(s,t)}{\cal H}_k^{(s,t)}=\lambda_k{\cal H}_k^{(s,t)},\quad k=1,2,\dots,m.
\end{align}
Equation \eqref{eigenpqe} shows that ${\cal H}_{1}^{(s,t)}$, ${\cal H}_{2}^{(s,t)}$, $\dots$, ${\cal H}_{m}^{(s,t)}$ are eigenvectors of $A^{(s,t)}$ corresponding to $\lambda_1$, $\lambda_2$, $\dots$, $\lambda_m$, respectively.\par

\end{proof}
\par
It is emphasized here that the discrete-time variable $t$ is meaningless for investigating the autonomous dToda equation \eqref{adToda}.
Without loss of generality, we can remove the discrete-time variable $t$ from the superscript in the case of the autonomous dToda equation \eqref{adToda}.
Thus, we henceforth regard $q_k^{(s,t)}, e_k^{(s,t)}$ in \eqref{q+en} with $t=0$ as $q_k^{(s)}, e_k^{(s)}$ in \eqref{adToda}, that is, $q_k^{(s)}=q_k^{(s,0)}$ and  $e_k^{(s)}=e_k^{(s,0)}$.
With the help of Proposition~\ref{prop4}, 
we derive the following theorem with respect to the solution to the autonomous dToda equation \eqref{adToda}.
%
%
\begin{theorem}\label{thm1}
The solution to the autonomous dToda equation \eqref{adToda} can be expressed as
\begin{align}
&q_k^{(s)}=\frac{H_{k-1}^{(s,0)}H_{k}^{(s+1,0)}}{H_{k}^{(s,0)}H_{k-1}^{(s+1,0)}},
\quad k=1,2,\dots,m,\label{qkstau}\\
&e_k^{(s)}=\frac{H_{k+1}^{(s,0)}H_{k-1}^{(s+1,0)}}{H_{k}^{(s,0)}H_{k}^{(s+1,0)}},\label{ekstau}
\quad k=1,2,\dots,m-1,
\end{align}
where $f_0^{(0)},f_{1}^{(0)},\dots, f_{m-1}^{(0)}$ appearing in the Hankel determinants $H_k^{(0,0)}$ 
are uniquely given from $q_1^{(0)}$, $q_2^{(0)}$, $\dots$, $q_m^{(0)}$ 
and $e_1^{(0)}$, $e_2^{(0)}$, $\dots$, $e_{m-1}^{(0)}$ in $A^{(0,0)}$. 
\end{theorem}
\begin{proof}
By ignoring the discrete-time variable $t$ in Lemma~\ref{lem3} and Theorem~\ref{lem5},
we derive \eqref{qkstau} and \eqref{ekstau} as the solution to the autonomous dToda equation \eqref{adToda}.
\par
Using a cofactor expansion along the $(k+1)$th column for the Hankel polynomials $H_k^{(s)}(z)$ without the discrete-time variable $t$ in \eqref{Tkz}, 
we obtain 
\begin{align}\label{thm1-1}
H_k^{(s)}(z)=H_k^{(s)}z^k+\tilde{H}_{k-1}^{(s)}z^{k-1}+\tilde{H}_{k-2}^{(s)}z^{k-2}+\dots+\tilde{H}_0^{(s)},
\end{align}
where $\tilde{H}_{k-i}^{(s)}$ denote determinants given by replacing the $(k-i+1)$th column with $-\bvec{f}_{k}^{(s+k)}$ 
in the Hankel determinants $H_k^{(s)}$, that is,
\begin{align*}
\tilde{H}_{k-i}^{(s)}:=\left|
\begin{array}{cccccccccccc}
\bvec{f}_{k}^{(s)}&\bvec{f}_k^{(s+1)}&
\cdots&\bvec{f}_k^{(s+k-i-1)}&
-\bvec{f}_k^{(s+k)}&\bvec{f}_k^{(s+k-i+1)}&
\cdots&\bvec{f}_k^{(s+k-1)}&
\end{array}\right|.
\end{align*}
Thus, by comparing \eqref{thm1-1} with ${\cal H}_{k}^{(s,t)}(z)=z^k+b_{k,1}^{(s,t)}z^{k-1}+b_{k,2}^{(s,t)}z^{k-2}+\cdots+b_{k,k}^{(s,t)}$, 
we derive $b_{k,1}^{(s)}=\tilde{H}_{k-1}^{(s)}/H_k^{(s)},b_{k,2}^{(s)}=\tilde{H}_{k-2}^{(s)}/H_k^{(s)},\dots,
b_{k,k}^{(s)}=\tilde{H}_{0}^{(s)}/H_k^{(s)}$.
Applying the inverse of Cramer's rule to these with $s=0$, we obtain
\begin{align}
\left(\begin{array}{cccc}
\bvec{f}_{k}^{(0)}&\bvec{f}_{k}^{(1)}&\cdots&\bvec{f}_{k}^{(k-1)}
\end{array}
\right)\left(
\begin{array}{c}
b_{k,k}^{(0)}\\b_{k,k-1}^{(0)}\\\vdots\\b_{k,1}^{(0)}
\end{array}\right)=
-\bvec{f}_{k}^{(k)},\quad
k=1,2,\dots,m.\label{beqs}
\end{align}
Therefore it follows that
\begin{align}
\left\{
\begin{array}{l}
f_0^{(0)}=1,\\
f_1^{(0)}=-b_{1,1}f_0^{(0)},\\
f_2^{(0)}=-b_{2,1}f_{1}^{(0)}-b_{2,2}f_0^{(0)},\\
\quad\vdots\\
f_{m-1}^{(0)}=
-b_{m-1,1}f_{m-2}^{(0)}-b_{m-1,2}f_{m-3}^{(0)}-\cdots-b_{m-1,m-1}f_{0}^{(0)},
\end{array}
\right.\label{fbdet}
\end{align}
which implies that $f_0^{(0)},f_{1}^{(0)},\dots, f_{m-1}^{(0)}$ are uniquely determined
using $b_{k,1}^{(0)}$, $b_{k,2}^{(0)}$, $\dots$, $b_{k,k}^{(0)}$.\par
Moreover, it is obvious from \eqref{phikstdet} with $s=0$ and without the discrete-time $t$ that 
$q_k^{(0)}$ and $e_k^{(0)}$ generate the Hadamard polynomial ${\cal H}_k^{(0)}(z)$.
Thus, we see that $b_{k,1}^{(0)}$, $b_{k,2}^{(0)}$, 
$\dots$, $b_{k,k}^{(0)}$ are uniquely given.

\end{proof}
\par
Similar to $L^{(s,t)}$ and $R^{(s,t)}$ involving $e_{k}^{(s,t)}$ and $q_k^{(s,t)}$, respectively,
let us define $m$-by-$m$ matrices involving $E_k^{(s,t)}$ and $Q_k^{(s,t)}$ appearing in Lemma~\ref{lem4} and Theorem~\ref{lem6}:
\begin{align*}
{\cal L}^{(s,t)}:=\left(
\begin{array}{cccc}
1&\\
E_1^{(s,t)}&1&\\
&\ddots&\ddots&\\
&&E_{m-1}^{(s,t)}&1
\end{array}
\right),\quad
{\cal R}^{(s,t)}:=\left(
\begin{array}{cccc}
Q_1^{(s,t)}&1\\
&Q_2^{(s,t)}&\ddots\\
&&\ddots&1\\
&&&Q_m^{(s,t)}
\end{array}
\right).
\end{align*}
Then, we have the following proposition concerning eigenpairs of ${\cal A}^{(s,t)}:={\cal L}^{(s,t)}{\cal R}^{(s,t)}$.
%
%
\begin{proposition}\label{prop5}
Eigenvalues of ${\cal A}^{(s,t)}+\mu^{(t)}I_m$ coincide with the roots $\lambda_1$, $\lambda_2$, $\dots$, $\lambda_m$ of the polynomial
$p(z)=(z-\lambda_1)(z-\lambda_2)\dots(z-\lambda_m)$.
Moreover, the eigenvectors corresponding to $\lambda_1,\lambda_2,\dots,\lambda_m$ 
are ${\cal H}_1^{(s,t)},{\cal H}_2^{(s,t)},\dots,{\cal H}_m^{(s,t)}$, respectively.
\end{proposition}
\begin{proof}
Let us prepare the $k$th principle submatrix of ${\cal A}^{(s,t)}$, namely,
\begin{align*}
{\cal A}_k^{(s,t)}:=
\left(
\begin{array}{cccc}
Q_1^{(s,t)}&1\\
Q_1^{(s,t)}E_1^{(s,t)}&Q_2^{(s,t)}+E_1^{(s,t)}&\ddots\\
&\ddots&\ddots&1\\
&&Q_{k-1}^{(s,t)}E_{k-1}^{(s,t)}&Q_k^{(s,t)}+E_{k-1}^{(s,t)}
\end{array}
\right).
\end{align*}
Then, by considering cofactor expansion along the $(k+1)$th row for
$\det[(z-\mu^{(t)})I_{k+1}-{\cal A}_{k+1}^{(s,t)}]$,
we obtain 
\begin{align}\label{detzmuicalA}
\det[(z-\mu^{(t)})I_{k+1}-{\cal A}_{k+1}^{(s,t)}]
&
=(z-Q_{k+1}^{(s,t)}-E_{k}^{(s,t)}-\mu^{(t)})
\det[(z-\mu^{(t)})I_{k}-{\cal A}_{k}^{(s,t)}]\nonumber\\
&
\quad\quad-Q_{k}^{(s,t)}E_{k}^{(s,t)}
\det[(z-\mu^{(t)})I_{k-1}-{\cal A}_{k-1}^{(s,t)}].
\end{align}
Comparing this with \eqref{3termQE},
we derive
\begin{align*}
{\cal H}_{k}^{(s,t)}=\det[(z-\mu^{(t)})I_k-{\cal A}_{k}^{(s,t)}],
\end{align*}
which implies $p(z)=\det[(z-\mu^{(t)})I_m-{\cal A}^{(s,t)}]$.
Thus, by recalling $p(z)=(z-\lambda_1)(z-\lambda_2)\cdots(z-\lambda_m)$,
we see that $\lambda_1,\lambda_2,\dots,\lambda_m$ are eigenvalues of ${\cal A}^{(s,t)}+\mu^{(t)}I_m$.
Similarly to the case of $q_k^{(s,t)}$ and $e_k^{(s,t)}$,
by reconsidering \eqref{3termQE} with $z=\lambda_k$,
we have
\begin{align}\label{eigenpairsofcalA}
({\cal A}^{(s,t)}+\mu^{(t)}I_m){\cal H}_k^{(s,t)}=\lambda_k{\cal H}_k^{(s,t)},\quad k=1,2,\dots,m.
\end{align}
Thus, ${\cal H}_{k}^{(s,t)}$ are eigenvectors of  ${\cal A}^{(s,t)}+\mu^{(t)}I_m$ corresponding to the eigenvalues $\lambda_k$.

\end{proof}
\par
With the help of Proposition~\ref{prop5},
we derive the following theorem concerning the solution to the non-autonomous dToda equation \eqref{nadToda}.
%
%
\begin{theorem}\label{thm2}
The solution to the non-autonomous dToda equation \eqref{nadToda} can be expressed as
\begin{align}
&Q_k^{(t)}=\frac{H_{k-1}^{(0,t)}H_{k}^{(0,t+1)}}{H_{k}^{(0,t)}H_{k-1}^{(0,t+1)}},
\quad k=1,2,\dots,m,\label{Qkttau}\\
&E_k^{(t)}=\frac{H_{k+1}^{(0,t)}H_{k-1}^{(0,t+1)}}{H_{k}^{(0,t)}H_{k}^{(0,t+1)}},
\quad k=1,2,\dots,m-1,\label{Ekttau}
\end{align}
where $f_0^{(0)},f_{1}^{(0)},\dots, f_{m-1}^{(0)}$ appearing in the Hankel determinants $H_k^{(0,0)}$ 
are uniquely given by $Q_1^{(0)}$, $Q_2^{(0)}$, $\dots$, $Q_m^{(0)}$ 
and $E_1^{(0)}$, $E_2^{(0)}$, $\dots$, $E_{m-1}^{(0)}$ in ${\cal A}^{(0,0)}$. 
\end{theorem}
\begin{proof}
The proof is essentially given using Lemma~\ref{lem4} and Theorem~\ref{lem6} instead of Lemma~\ref{lem3} and Theorem~\ref{lem5}, and replacing $s$ with $t$.

\end{proof}

\par
It is obvious that the solutions $q_k^{(s)}, e_k^{(s)}$ in Theorem~\ref{thm1} or $Q_k^{(t)},E_k^{(t)}$ in Theorem~\ref{thm2} have sufficient arbitrariness from the constants $a_1,a_2,\dots,a_m$ and the initial values $f_0^{(0)}, f_1^{(0)},\dots,f_{m-1}^{(0)}$ of the infinite sequence $\{f_s^{(t)}\}_{s,t=0}^{\infty}$.
From the perspective of the eigenvalue problem,
we also discuss the arbitrariness of $q_k^{(s)}, e_k^{(s)}$ or $Q_k^{(t)},E_k^{(t)}$.
If $q_k^{(0,0)}, e_k^{(0,0)}$ or $Q_k^{(0,0)},E_k^{(0,0)}$ are given,
we uniquely obtain $A_k^{(0,0)}$ or ${\cal A}_k^{(0,0)}$.
In other words, the eigenvalues $\lambda_1,\lambda_2,\dots,\lambda_m$ of both $A^{(0,0)}$ and ${\cal A}^{(0,0)}+\mu^{(0)}I_m$
are determined by $q_k^{(0,0)}, e_k^{(0,0)}$ or $Q_k^{(0,0)},E_k^{(0,0)}$.
Thus, the eigenvalues $\lambda_1,\lambda_2,\dots,\lambda_m$ have sufficient arbitrariness.
\section{Asymptotic behavior of discrete Toda equations}\label{sec:5}
In this section, we find a general term to the infinite sequence $\{f_{s}^{(t)}\}_{s,t=0}^{\infty}$
associated with distinct $\lambda_1,\lambda_2,\dots,\lambda_m$, 
and then derive asymptotic expansions of the Hankel determinants $H_k^{(s,t)}$
involving $f_s^{(t)}$ as $s\rightarrow\infty$ or $t\rightarrow\infty$.
Using the asymptotic expansions,
we clarify asymptotic behavior of the determinantal solutions to the autonomous Toda equation \eqref{adToda} 
and the non-autonomous dToda equation \eqref{nadToda}.
\par
The following lemma gives a general term to  the infinite sequence $\{f_{s}^{(t)}\}_{s=0}^{\infty}$ at fixed $t$.
%
%
\begin{lemma}\label{lem7}
Let us assume that the elements of the infinite sequence $\{f_{s}^{(t)}\}_{s=0}^{\infty}$ 
satisfy the linear dependence \eqref{fnmt} at fixed $t$.
Then, $f_{s}^{(t)}$ can be expressed using distinct $\lambda_1,\lambda_2,\dots,\lambda_m$ as
\begin{align}\label{gsfst}
f_{s}^{(t)}=\sum_{\ell=1}^{m}c_{\ell}^{(t)}\lambda_{\ell}^{s},\quad s=0,1,\dots,
\end{align}
where $c_1^{(t)},c_2^{(t)},\dots,c_m^{(t)}$ are constants determined by
\begin{align}\label{cdetermine}
\left(
\begin{array}{c}
c_{1}^{(t)}\\c_{2}^{(t)}\\\vdots\\c_{m}^{(t)}\\
\end{array}
\right)=
\left(
\begin{array}{cccc}
1&1&\cdots&1\\
\lambda_1&\lambda_2&\cdots&\lambda_m\\
\vdots&\vdots&\ddots&\vdots\\
\lambda_1^{m-1}&\lambda_{2}^{m-1}&\cdots&\lambda_{m}^{m-1}
\end{array}
\right)^{\hspace{-2mm}-1}
\left(
\begin{array}{c}
f_{0}^{(t)}\\f_{1}^{(t)}\\\vdots\\f_{m-1}^{(t)}\\
\end{array}
\right),
\end{align}
and satisfying
\begin{align}
c_i^{(t)}=c_i^{(0)}\prod_{\ell=0}^{t-1}(\lambda_{i}-\mu^{(\ell)}),\quad i=1,2,\dots,m.\label{citoi0}
\end{align}
\end{lemma}
\begin{proof}
Using \eqref{gsfst}, we can rewrite the left-hand side of the linear dependence \eqref{fnmt} as
\begin{align*}
\sum_{i=0}^{m}a_{i}f_{s+m-i}^{(t)}
&=
\sum_{i=0}^{m}a_{i}\left(
\sum_{\ell=1}^{m}c_{\ell}^{(t)}\lambda_{\ell}^{s+m-i}\right)\\
&=\sum_{\ell=1}^m
c_{\ell}^{(t)}\lambda_\ell^{s}\left(
\sum_{i=1}^{m}a_{i}\lambda_{\ell}^{m-i}
\right).
\end{align*}
From \eqref{pza}, it is obvious that 
$\sum_{i=1}^{m}a_{i}\lambda_{\ell}^{m-i}=p(\lambda_\ell)$.
By recalling that $\lambda_\ell$ are zeros of $p(z)$,
we derive $\sum_{i=1}^{m}a_{i}\lambda_{\ell}^{m-i}=p(\lambda_\ell)=0$.
Thus, we see that $f_s^{(t)}$ in \eqref{gsfst} satisfy \eqref{fnmt}.
From \eqref{gsfst}, it follows that
\begin{align}\label{coefmatrix}
\left(\begin{array}{c}
f_0^{(t)}\\f_1^{(t)}\\\vdots\\f_{m-1}^{(t)}
\end{array}\right)=
\left(\begin{array}{cccc}
1&1&\cdots&1\\
\lambda_1&\lambda_2&\cdots&\lambda_m\\
\vdots&\vdots&\ddots&\vdots\\
\lambda_1^{m-1}&\lambda_{2}^{m-1}&\cdots&\lambda_{m}^{m-1}
\end{array}
\right)\left(
\begin{array}{c}
c_{1}^{(t)}\\c_{2}^{(t)}\\\vdots\\c_{m}^{(t)}\\
\end{array}
\right).
\end{align}
Then, by considering that the $m$-by-$m$ matrix involving $\lambda_1,\lambda_2,\dots,\lambda_m$ 
is the Vandermonde matrix with determinant $\prod_{i<j}(\lambda_j-\lambda_i)$, we obtain \eqref{cdetermine}.
Moreover, by combining \eqref{key} in Proposition~\ref{prop2} with \eqref{gsfst},
we derive 
\begin{align}\label{cit+1}
c_i^{(t+1)}=(\lambda_i-\mu^{(t)})c_i^{(t)},\quad i=1,2,\dots,m,
\end{align}
which immediately leads to \eqref{citoi0}.

\end{proof}
\par
It is of significance that the infinite sequence $\{f_s^{(t)}\}_{s=0}^{\infty}$ at any $t$ in Lemma~\ref{lem7}
involves arbitrary constants $\lambda_1,\lambda_2,\dots,\lambda_m$ and initial values $f_0^{(t)}, f_1^{(t)},\dots,f_{m-1}^{(t)}$.
Thus, the infinite sequence $\{f_s^{(t)}\}_{s=0}^{\infty}$ has $2m$ degrees of freedom,
which is greater than in the initial setting of the autonomous dToda equation \eqref{adToda}. 
In other words, the freedom in the infinite sequence $\{f_s^{(t)}\}_{s=\infty}^{\infty}$ sufficiently covers 
that in the setting of the autonomous dToda equation \eqref{adToda}.
\par
The following lemma gives an expansion of the Hankel determinants $H_k^{(s,t)}$.
%
%
\begin{lemma}\label{lem8}
At any $t$, the Hankel determinants $H_k^{(s,t)}$ can be expanded as
\begin{align}\label{aeHankels}
H_k^{(s,t)}=\sum_{1\le \kappa_1<\kappa_2<\dots<\kappa_k\le m}
c_{\kappa_1}^{(t)}c_{\kappa_2}^{(t)}\cdots c_{\kappa_k}^{(t)} 
(\lambda_{\kappa_1}\lambda_{\kappa_2}\cdots \lambda_{\kappa_k})^s
\prod_{\ell=2}^{k}\lambda_{\kappa_\ell}^{\ell-1}
\prod_{i<j}(\lambda_{\kappa_j}-\lambda_{\kappa_i}).
\end{align}
\end{lemma}
\begin{proof}
Since the entries of the Hankel determinants $H_k^{(s,t)}$ satisfy \eqref{gsfst} in Lemma~\ref{lem7},
we derive 
\begin{align*}
H_k^{(s,t)}=\left|
\begin{array}{cccc}
\displaystyle\sum_{\ell=1}^{m}c_{\ell}^{(t)}\lambda_{\ell}^{s}
&\displaystyle\sum_{\ell=1}^{m}c_{\ell}^{(t)}\lambda_{\ell}^{s+1}
&\cdots
&\displaystyle\sum_{\ell=1}^{m}c_{\ell}^{(t)}\lambda_{\ell}^{s+k-1}\\
\displaystyle\sum_{\ell=1}^{m}c_{\ell}^{(t)}\lambda_{\ell}^{s+1}
&\displaystyle\sum_{\ell=1}^{m}c_{\ell}^{(t)}\lambda_{\ell}^{s+2}
&\cdots
&\displaystyle\sum_{\ell=1}^{m}c_{\ell}^{(t)}\lambda_{\ell}^{s+k}\\
\vdots&\vdots&\ddots&\vdots\\
\displaystyle\sum_{\ell=1}^{m}c_{\ell}^{(t)}\lambda_{\ell}^{s+k-1}
&\displaystyle\sum_{\ell=1}^{m}c_{\ell}^{(t)}\lambda_{\ell}^{s+k}
&\cdots
&\displaystyle\sum_{\ell=1}^{m}c_{\ell}^{(t)}\lambda_{\ell}^{s+2k-2}\\
\end{array}
\right|,\quad k=1,2,\dots,m.
\end{align*}
By reorganizing these determinants, we obtain
\begin{align}\label{taukccclllv}
&H_k^{(s,t)}=
\sum_{1\le \kappa_1<\kappa_2<\dots<\kappa_k\le m}
c_{\kappa_1}^{(t)}c_{\kappa_2}^{(t)}\cdots c_{\kappa_k}^{(t)} 
(\lambda_{\kappa_1}\lambda_{\kappa_2}\cdots \lambda_{\kappa_k})^s
\left|
\begin{array}{cccc}
1&\lambda_{\kappa_2}&\cdots&\lambda_{\kappa_k}^{k-1}\\
\lambda_{\kappa_1}&\lambda_{\kappa_2}^2&\cdots&\lambda_{\kappa_k}^k\\
\vdots&\vdots&\ddots&\vdots\\
\lambda_{\kappa_1}^{k-1}&\lambda_{\kappa_2}^k&\cdots&\lambda_{\kappa_k}^{2k-2}
\end{array}
\right|,\nonumber\\
&\quad
k=1,2,\dots,m.
\end{align}
Noting that the Vandermonde determinants appear on the right hand side of \eqref{taukccclllv},
we thus have \eqref{aeHankels}.

\end{proof}
\par
Lemma~\ref{lem8} immediately leads to asymptotic expansions of the Hankel determinants $H_k^{(s,t)}$ as $s\rightarrow\infty$.
%
%
\begin{lemma}\label{lem9}
Let us assume that $\lambda_1,\lambda_2,\dots,\lambda_m$ are constants 
such that $|\lambda_1|>|\lambda_2|>\dots>|\lambda_m|$.
Then, the Hankel determinants $H_k^{(s,t)}$ can be expanded as
\begin{align}\label{astauks}
H_k^{(s,t)}=\hat{c}_{k}^{(t)} (\lambda_1\lambda_2\cdots\lambda_k)^s\left(1+O\left(\rho_k^s\right)\right),
\quad s\rightarrow\infty,
\end{align}
where $\hat{c}_k^{(t)}:=c_1^{(t)}c_2^{(t)}\cdots c_k^{(t)}\prod_{\ell=2}^{k}\lambda_{\kappa_\ell}^{\ell-1}
\prod_{i<j}(\lambda_{\kappa_j}-\lambda_{\kappa_i})$
and $\rho_k$ are some constants such that $\rho_k>|\lambda_{k+1}|/|\lambda_{k}|$.
\end{lemma}
\begin{proof}
Under the assumption $|\lambda_1|>|\lambda_2|>\dots>|\lambda_m|$,
the dominant terms of the Hankel determinants $H_k^{(s,t)}$ as $s\rightarrow\infty$ 
are the terms with  $\kappa_1=1,\kappa_2=2,\dots,\kappa_k=k$ appearing in Lemma~\ref{lem8}.
\end{proof}
\par
With the help of Lemma~\ref{lem9}, 
we present an asymptotic analysis of the autonomous dToda equation \eqref{adToda} 
as $s\rightarrow\infty$.
%
%
\begin{theorem}\label{thm3}
Let us assume that for the infinite sequence $\{f_s\}_{s=0}^{\infty}$, the Hankel determinants $H_k^{(s)}$ are all nonzero.
Moreover, let $\lambda_1,\lambda_2,\dots,\lambda_m$ 
be constants such that $|\lambda_1|>|\lambda_2|>\dots>|\lambda_m|$.
Then, the asymptotic behavior as $s\rightarrow\infty$ of the autonomous dToda variables $q_k^{(s)}$ and $e_k^{(s)}$ are given as
 \begin{align}
&\lim_{s\rightarrow\infty}
q_{k}^{(s)}=\lambda_k,\quad k=1,2,\dots,m,\label{qkstsinfty}\\
&\lim_{s\rightarrow\infty}
e_k^{(s)}=0,\quad k=1,2,\dots, m-1.\label{ekstsinfty}
\end{align}
\end{theorem}
\begin{proof}
By applying Lemma~\ref{lem9} to the expressions of $q_{k}^{(s,t)}$ and $e_k^{(s,t)}$ in Lemma~\ref{lem3},
we obtain, as $s\rightarrow\infty$,
\begin{align}
&q_k^{(s)}=\lambda_k \frac{\left(1+O\left(\rho_{k-1}^s\right)\right)\left(1+O\left(\rho_k^{s+1}\right)\right)}
{\left(1+O\left(\rho_k^s\right)\right)\left(1+O\left(\rho_{k-1}^{s+1}\right)\right)},
\quad k=1,2,\dots,m,\label{qkstOs}\\
&e_k^{(s)}=\frac{\hat{c}_{k+1}^{(t)}\hat{c}_{k-1}^{(t)}}{(\hat{c}_k^{(t)})^2}
\frac{1}{\lambda_k}
\left(\frac{\lambda_{k+1}}{\lambda_k}
\right)^{\hspace{-1mm}s}
\frac{\left(1+O\left(\rho_{k+1}^s\right)\right)\left(1+O\left(\rho_{k-1}^{s+1}\right)\right)}
{\left(1+O\left(\rho_k^s\right)\right)\left(1+O\left(\rho_{k}^{s+1}\right)\right)},
\quad k=1,2,\dots,m-1.\label{ekstOs}
\end{align}
Thus, by taking the limit $s\rightarrow\infty$ in \eqref{qkstOs} and \eqref{ekstOs},
we see that $q_k^{(s,t)}\rightarrow\lambda_k$ and $e_k^{(s,t)}\rightarrow 0$ as $s\rightarrow\infty$.
We may remove the discrete-time variable $t$ from the superscripts in $q_k^{(s,t)}$ and $e_k^{(s,t)}$
in the case of analysis for the autonomous dToda equation \eqref{adToda}.
Therefore, we have \eqref{qkstsinfty} and \eqref{ekstsinfty}

\end{proof}
\par
The discrete-time evolution from $s$ to $s+1$ in $Q_k^{(s,t)}$ and $E_k^{(s,t)}$ 
is not considered in either the autonomous dToda equation \eqref{adToda} and the non-autonomous dToda equation \eqref{nadToda}.
Although the explicit formula for generating such a discrete-time evolution has not been shown,
we can describe an asymptotic behavior of $Q_k^{(s,t)}$ and $E_k^{(s,t)}$ as $s\rightarrow\infty$.
\begin{theorem}\label{thm4}
Let us assume that for the  infinite sequence $\{f_s^{(t)}\}_{s,t=0}^{\infty}$, the Hankel determinants $H_k^{(s,t)}$ are all nonzero.
Moreover, let $\lambda_1,\lambda_2,\dots,\lambda_m$ be constants such that $|\lambda_1|>|\lambda_2|>\dots>|\lambda_m|$.
Then, it holds that, as $s\rightarrow\infty$, 
\begin{align}
&\lim_{s\rightarrow\infty}Q_k^{(s,t)}=\lambda_{k}-\mu^{(t)},\quad k=1,2,\dots,m,\\
&\lim_{s\rightarrow\infty}E_k^{(s,t)}=0,\quad k=1,2,\dots,m-1.
\end{align}
\end{theorem}
\begin{proof}
For the expressions of $Q_k^{(s,t)}$ and $E_k^{(s,t)}$ in Lemma~\ref{lem4} as $s\rightarrow\infty$,
it follows that, as $s\rightarrow\infty$,
\begin{align}
&Q_k^{(s,t)}=\frac{\hat{c}_{k-1}^{(t)}\hat{c}_{k}^{(t+1)}}{\hat{c}_k^{(t)}\hat{c}_{k-1}^{(t+1)}}
\frac{\left(1+O\left(\rho_{k-1}^s\right)\right)\left(1+O\left(\rho_k^{s}\right)\right)}
{\left(1+O\left(\rho_k^s\right)\right)\left(1+O\left(\rho_{k-1}^{s}\right)\right)},
\quad k=1,2,\dots,m,\label{QkstOs}\\
&E_k^{(s,t)}=
\frac{\hat{c}_{k+1}^{(t)}\hat{c}_{k-1}^{(t+1)}}{\hat{c}_{k}^{(t)}\hat{c}_k^{(t+1)}}
\left(\frac{\lambda_{k+1}}{\lambda_k}
\right)^{\hspace{-1mm}s}
\frac{\left(1+O\left(\rho_{k+1}^s\right)\right)\left(1+O\left(\rho_{k-1}^{s}\right)\right)}
{\left(1+O\left(\rho_k^s\right)\right)\left(1+O\left(\rho_{k}^{s}\right)\right)}
\quad k=1,2,\dots,m-1.\label{EkstOs}
\end{align}
Equations \eqref{QkstOs} and \eqref{EkstOs}
imply that, as $s\rightarrow\infty$, $Q_k^{(s,t)}\rightarrow\hat{c}_{k-1}^{(t)}\hat{c}_{k}^{(t+1)}/(\hat{c}_{k}^{(t)}\hat{c}_{k-1}^{(t+1)})$ 
and $E_k^{(s,t)}\rightarrow0$, respectively.
Noting that $\hat{c}_{k}^{(t+1)}/\hat{c}_{k}^{(t)}
=\prod_{\ell=1}^{k}c_{\ell}^{(t+1)}/c_{\ell}^{(t)}
=\prod_{\ell=1}^{k}(\lambda_{\ell}-\mu^{(t)})$,
we also derive $Q_k^{(s,t)}\rightarrow\lambda_{k}-\mu^{(t)}$ as $s\rightarrow\infty$.

\end{proof}
\par
Theorems~\ref{thm3} and~\ref{thm4} with
Propositions~\ref{prop4} and~\ref{prop5} claim that, for fixed $t$, $q_k^{(s)}=q_k^{(s,t)}$ and $Q_k^{(s)}=Q_k^{(s,t)}$ converge to
the eigenvalues of $A^{(s)}=A^{(s,t)}$ and ${\cal A}^{(s)}={\cal A}^{(s,t)}$ as $s\rightarrow\infty$, respectively.
The asymptotic convergence of $q_k^{(s)}$ and $e_k^{(s)}$ as $s\rightarrow\infty$ obviously coincides with that shown in \cite{Henrici,Rutishauser1990}.
\par
Let us turn to the asymptotic analysis of $q_k^{(s,t)}, e_k^{(s,t)}$ and $Q_k^{(s,t)}, E_k^{(s,t)}$ as $t\rightarrow\infty$.
The following lemma gives asymptotic expansions of the Hankel determinants $H_k^{(s,t)}$ as $t\rightarrow\infty$.
%
%
\begin{lemma}\label{lem10}
Let us assume that $\lambda_1,\lambda_2,\dots,\lambda_m$ are constants such that 
$|\lambda_1-\mu^{(t)}|>|\lambda_2-\mu^{(t)}|>\dots>|\lambda_m-\mu^{(t)}|$.
Then, as $t\rightarrow\infty$,
\begin{align}\label{asympexpantaukt}
&H_k^{(s,t)}=
\check{c}_k^{(s)}
\prod_{i=0}^{t}
(\lambda_{1}-\mu^{(i)})(\lambda_{2}-\mu^{(i)})\cdots(\lambda_{k}-\mu^{(i)})
\left(
1+O\left(\varrho_k^t\right)
\right),
\quad k=1,2,\dots,m,
\end{align}
where $\check{c}_{k}^{(s)}$ are constants given by
$\check{c}_k^{(s)}:=c_1^{(0)}c_2^{(0)}\cdots c_k^{(0)}(\lambda_1\lambda_2\cdots \lambda_k)^s$
$\prod_{\ell=2}^{k}\lambda_{\kappa_\ell}^{\ell-1}$
$\prod_{i<j}(\lambda_{\kappa_j}-\lambda_{\kappa_i})$
and $\varrho_k$ are constants satisfying $\varrho_{k}>|\lambda_{k+1}-\mu^{(t)}|/|\lambda_{k}-\mu^{(t)}|$.
\end{lemma}
\begin{proof}
Equation \eqref{asympexpantaukt} is easily proved by noting that the terms with
$\kappa_1=1,\kappa_2=2,\dots,\kappa_k=k$ appearing in Lemma~\ref{lem8}
are dominant terms of the Hankel determinants $H_k^{(s,t)}$ as $t\rightarrow\infty$
under the assumption $|\lambda_1-\mu^{(t)}|>|\lambda_2-\mu^{(t)}|>\dots>|\lambda_m-\mu^{(t)}|$.

\end{proof}
\par
Neither the autonomous dToda equation \eqref{adToda} nor the non-autonomous dToda equation \eqref{nadToda}
generates the discrete-time evolution of $q_k^{(s,t)}$ and $e_k^{(s,t)}$ with respect to $t$.
However, Lemma~\ref{lem10} immediately yields asymptotic behavior of $q_k^{(t)}$ and $e_k^{(t)}$ obtained by removing the discrete-time variable $s$ from the superscripts in $q_k^{(s,t)}$ and $e_k^{(s,t)}$, respectively, as $t\rightarrow\infty$. 
%
%
\begin{theorem}\label{thm5}
Let us assume that for the infinite sequence $\{f^{(t)}\}_{t=0}^{\infty}$, 
the Hankel determinants  $H_k^{(s,t)}$ are all nonzero.
Moreover, let $\lambda_1,\lambda_2,\dots,\lambda_m$ be constants 
such that $|\lambda_1-\mu^{(t)}|>|\lambda_2-\mu^{(t)}|>\dots>|\lambda_m-\mu^{(t)}|$.
Then, it holds that
\begin{align}
&\lim_{t\rightarrow\infty} q_k^{(t)}
=\lambda_k,\quad k=1,2,\dots,m,\\
&\lim_{t\rightarrow\infty} e_k^{(t)}
=0,\quad k=1,2,\dots,m-1.
\end{align}
\end{theorem}
\begin{proof}
By combining Lemma~\ref{lem10} with the expressions of $q_k^{(s,t)}$ and $e_{k}^{(s,t)}$ in Lemma~\ref{lem3},
we obtain, as $t\rightarrow\infty$,
\begin{align}
&q_k^{(s,t)}=\frac{\check{c}_{k-1}^{(s)}\check{c}_k^{(s+1)}}{\check{c}_k^{(s)}\check{c}_{k-1}^{(s+1)}}
\frac{\left(1+O\left(\varrho_{k-1}^{t}\right)\right)
\left(1+O\left(\varrho_{k}^{t}\right)\right)}
{\left(1+O\left(\varrho_{k}^{t}\right)\right)
\left(1+O\left(\varrho_{k-1}^{t}\right)\right)},\quad
k=1,2,\dots,m,\label{qkstOt}\\
&e_k^{(s,t)}=\frac{\check{c}_{k+1}^{(s)}\check{c}_{k-1}^{(s+1)}}{\check{c}_k^{(s)}\check{c}_k^{(s+1)}}
\prod_{i=1}^{t}\frac{\lambda_{k+1}-\mu^{(i)}}{\lambda_k-\mu^{(i)}}
\frac{\left(1+O\left(\varrho_{k}^{t}\right)\right)
\left(1+O\left(\varrho_{k-1}^{t}\right)\right)}
{\left(1+O\left(\varrho_{k}^{t}\right)\right)
\left(1+O\left(\varrho_{k}^{t}\right)\right)},
\quad k=1,2,\dots,m-1.\label{ekstOt}
\end{align}
From the limit $t\rightarrow\infty$ in \eqref{qkstOt} and \eqref{ekstOt},
we find that $q_k^{(s,t)}
\rightarrow
\check{c}_{k-1}^{(s)}\check{c}_k^{(s+1)}/$
$(\check{c}_k^{(s)}\check{c}_{k-1}^{(s+1)})
=\lambda_k$ and  $e_k^{(s,t)}\rightarrow 0$ as $t\rightarrow\infty$.

\end{proof}
\par
Lemma~\ref{lem10} also enables us to clarify the asymptotic analysis of the non-autonomous dToda equation \eqref{nadToda}.
%
%
\begin{theorem}\label{thm6}
Let us assume that for the infinite sequence $\{f^{(t)}\}_{t=0}^{\infty}$, 
the Hankel determinants  $H_k^{(s,t)}$ are all nonzero.
Moreover, let $\lambda_1,\lambda_2,\dots,\lambda_m$ be constants 
such that $|\lambda_1-\mu^{(t)}|>|\lambda_2-\mu^{(t)}|>\dots>|\lambda_m-\mu^{(t)}|$.
Then, it holds that
\begin{align}
&\lim_{t\rightarrow\infty} Q_k^{(t)}=\lambda_k-\mu^\ast,\quad k=1,2,\dots,m,\\
&\lim_{t\rightarrow\infty} E_k^{(t)}=0,\quad k=1,2,\dots,m-1,
\end{align}
where $\mu^{\ast}:=\lim_{t\rightarrow\infty}\mu^{(t)}$.
\end{theorem}
\begin{proof}
From the expressions of $Q_k^{(s,t)}$ and $E_k^{(s,t)}$ in Lemma~\ref{lem4} and Lemma~\ref{lem10},
it follows that, as $t\rightarrow\infty$,
\begin{align*}
&Q_k^{(s,t)}=(\lambda_k-\mu^{(t+1)})
\frac{\left(1+O\left(\varrho_{k-1}^{t}\right)\right)
\left(1+O\left(\varrho_{k}^{t+1}\right)\right)}
{\left(1+O\left(\varrho_{k}^{t}\right)\right)
\left(1+O\left(\varrho_{k-1}^{t+1}\right)\right)},
\quad k=1,2,\dots,m,\\
&E_k^{(s,t)}=\frac{\check{c}_{k+1}^{(s)}\check{c}_{k-1}^{(s)}}{(\check{c}_k^{(s)})^2}
\frac{1}{\lambda_k-\mu^{(t+1)}}\prod_{i=1}^{t}\frac{\lambda_{k+1}-\mu^{(i)}}{\lambda_k-\mu^{(i)}}
\frac{\left(1+O\left(\varrho_{k+1}^{t}\right)\right)
\left(1+O\left(\varrho_{k-1}^{t+1}\right)\right)}
{\left(1+O\left(\varrho_{k}^{t}\right)\right)
\left(1+O\left(\varrho_{k}^{t+1}\right)\right)},
\quad k=1,2,\dots,m-1,
\end{align*}
which imply that $Q_k^{(s,t)}\rightarrow \lambda_k-\mu^{\ast}$ and $E_k^{(s,t)}\rightarrow 0$ as $t\rightarrow\infty$.

\end{proof}
\par
From Theorems~\ref{thm5} and~\ref{thm6} with Propositions~\ref{prop4} and~\ref{prop5},
it can be concluded that $q_k^{(t)}$ and  $Q_k^{(t)}$ converge to the eigenvalues of $A^{(t)}$
and their shifted values of ${\cal A}^{(t)}$, as $t\rightarrow\infty$, respectively.
In particular,
Theorem~\ref{thm6} with Proposition~\ref{prop5} also describes asymptotic convergence to matrix eigenvalues 
in the so-called qd with implicit shift algorithm.

\section{Discrete Lotka-Volterra system, its determinant solution and asymptotic behavior}\label{sec:6}
In this section, we first derive the dLV system \eqref{dLV} by observing properties concerning the Hadamard polynomials ${\cal H}_k^{(s,t)}(z)$ involving the infinite sequence $\{f_{s}^{(t)}\}_{s,t=0}^{\infty}$ again.
Then, we show the determinant solution to the dLV system \eqref{dLV}, and give an asymptotic analysis for the dLV system  \eqref{dLV}.\par
For the Hadamard polynomials ${\cal H}_k^{(s,t)}(z)$, 
let us prepare symmetrical polynomials $\tilde{\cal H}_k^{(s,t)}(z)$ given by
\begin{align}\label{psiphi}
\left\{
\begin{array}{l}
\tilde{\cal H}_{-1}^{(s,t)}(z):=0,\\
\tilde{\cal H}_{2k}^{(s,t)}(z):={\cal H}_k^{(s,t)}(z^2),\quad k=0,1,\dots,m,\\
\tilde{\cal H}_{2k+1}^{(s,t)}(z):=z{\cal H}_{k}^{(s+1,t)}(z^2),\quad k=0,1,\dots,m.
\end{array}
\right.
\end{align}
Hereinafter, we refer to the symmetric polynomials $\tilde{\cal H}_k^{(s,t)}(z)$ as the symmetric Hadamard polynomials.
\par
Similar to the case of the Hadamard polynomials ${\cal H}_k^{(s,t)}(z)$,
we easily derive three-term recurrence relations with respect to 
the symmetric Hadamard polynomials $\tilde{\cal H}_{k}^{(s,t)}(z)$.
%
%
\begin{lemma}\label{lem11}
The symmetric Hadamard polynomials $\tilde{\cal H}_k^{(s,t)}(z)$ satisfy
\begin{align}\label{psi3term}
z\tilde{\cal H}_k^{(s,t)}(z)=\tilde{\cal H}_{k+1}^{(s,t)}(z)+v_k^{(s,t)}\tilde{\cal H}_{k-1}^{(s,t)}(z),\quad
k=0,1,\dots,2m,
\end{align}
where $v_{2k}^{(s,t)}:=e_k^{(s,t)}$ and $v_{2k-1}^{(s,t)}:=q_k^{(s,t)}$.
\end{lemma}
\begin{proof}
By replacing $z$ with $z^2$ in Lemma~\ref{lem3},
we obtain 
\begin{align*}
&z^2{\cal H}_{k-1}^{(s+1,t)}(z^2)={\cal H}_{k}^{(s,t)}(z^2)+q_{k}^{(s,t)}{\cal H}_{k-1}^{(s,t)}(z^2),\quad k=1,2,\dots,m,\\
&{\cal H}_k^{(s,t)}(z^2)={\cal H}_{k}^{(s+1,t)}+e_k^{(s,t)}{\cal H}_{k-1}^{(s+1,t)}(z^2),\quad k=0,1,\dots,m.
\end{align*}
Thus, by considering \eqref{psiphi}, we have \eqref{psi3term}.

\end{proof}
\par
For the parameters $\mu^{(t)}$ in the non-autonomous dToda equation \eqref{nadToda}, 
let us introduce an infinite sequence $\{\kappa^{(t)}\}_{t=0}^{\infty}$ satisfying
\begin{align}\label{kappamu}
(\kappa^{(t)})^2:=\mu^{(t)},\quad t=0,1,\dots.
\end{align}
Then, we derive another three-term recurrence relations of the symmetric Hadamard polynomials
$\tilde{\cal H}_k^{(s,t)}(z)$ involving the infinte sequence $\{\kappa^{(t)}\}_{t=0}^{\infty}$.
%
%
\begin{lemma}\label{lem12}
The symmetric Hadamard polynomials $\tilde{\cal H}_{k}^{(s,t)}(z)$ satisfy
\begin{align}\label{psichrist}
[z^2-(\kappa^{(t)})^2]
\tilde{\cal H}_{k-1}^{(s,t+1)}(z)
=\tilde{\cal H}_{k+1}^{(s,t)}(z)+V_{k}^{(s,t)}\tilde{\cal H}_{k-1}^{(s,t)}(z),\quad k=1,2,\dots,2m,
\end{align}
where $V_{2k-1}^{(s,t)}:=Q_{k}^{(s,t)}$ and $V_{2k}^{(s,t)}:=Q_{k}^{(s+1,t)}$.
\end{lemma}
\begin{proof}
From \eqref{Christoffelt} in Lemma~\ref{lem4},
it is obvious that
\begin{align*}
&[z^2-(\kappa^{(t)})^2]\tilde{\cal H}_{2k-2}^{(s,t+1)}(z)
=\tilde{\cal H}_{2k}^{(s,t)}(z)+Q_{k}^{(s,t)}\tilde{\cal H}_{2k-2}^{(s,t)}(z),\\
&[z^2-(\kappa^{(t)})^2]\tilde{\cal H}_{2k-1}^{(s,t+1)}(z)
=\tilde{\cal H}_{2k+1}^{(s,t)}(z)+Q_{k}^{(s+1,t)}\tilde{\cal H}_{2k-1}^{(s,t)}(z),
\end{align*} 
which are equivalent to \eqref{psichrist}.

\end{proof}
Lemmas~\ref{lem11} and~\ref{lem12} lead to the following lemma with respect to $v_k^{(s,t)}$ and $V_k^{(s,t)}$.
%
%
\begin{lemma}
The variables $v_k^{(s,t)}$ and $V_k^{(s,t)}$ satisfy
\begin{align}
&v_k^{(s,t+1)}+V_{k+2}^{(s,t)}=v_{k+2}^{(s,t)}+V_{k+1}^{(s,t)},
\quad k=0,1,\dots,2m-2,\label{v+V}\\
&v_{k}^{(s,t+1)}V_{k}^{(s,t)}=v_{k}^{(s,t)}V_{k+1}^{(s,t)},
\quad k=1,2,\dots,2m-1. \label{vV}
\end{align}
\end{lemma}
\begin{proof}
From Lemmas~\ref{lem11} and~\ref{lem12},
it follows that
\begin{align*}
(v_{k}^{(s,t+1)}+V_{k+2}^{(s,t)}-v_{k+2}^{(s,t)}-V_{k+1}^{(s,t)})\tilde{\cal H}_{k+1}^{(s,t)}(z)
+(v_{k}^{(s,t+1)}V_{k}^{(s,t)}-v_{k}^{(s,t)}V_{k+1}^{(s,t)})\tilde{\cal H}_{k-1}^{(s,t)}(z)=0.
\end{align*}
It is obvious that the symmetric Hadamard polynomials $\tilde{\cal H}_{k+1}^{(s,t)}(z)$ and $\tilde{\cal H}_{k-1}^{(s,t)}(z)$  
are linear independent.
Thus, we have \eqref{v+V} and \eqref{vV}.
\end{proof}
%
%
\begin{lemma}\label{lem14}
Let $u_k^{(s,t)}:=v_k^{(s,t)}(\tilde{\cal H}_{k-1}^{(s,t)}(\kappa^{(t)})/\tilde{\cal H}_{k}^{(s,t)}(\kappa^{(t)}))$ for $k=0,1,\dots,2m$.
Then, $v_k^{(s,t)}$ and $V_k^{(s,t)}$ can be expressed in terms of $u_k^{(s,t)}$ as
\begin{align}
&v_k^{(s,t)}=u_k^{(s,t)}(\kappa^{(t)}-u_{k-1}^{(s,t)}),\quad k=1,2,\dots,2m,\label{vu}\\
&V_k^{(s,t)}=-(\kappa^{(t)}-u_{k-1}^{(s,t)})
(\kappa^{(t)}-u_{k}^{(s,t)}),\quad k=1,2,\dots,2m.\label{Vu}
\end{align} 
Moreover, it holds that
\begin{align}
u_k^{(s,t+1)}(\kappa^{(t+1)}-u_{k-1}^{(s,t+1)})
=
u_{k}^{(s,t)}(\kappa^{(t)}-u_{k+1}^{(s,t)}),\quad
k=1,2,\dots,2m-1.\label{odLV}
\end{align}
\end{lemma}
\begin{proof}
By letting $z=\kappa^{(t)}$ in Lemma~\ref{lem11}
and by replacing $\tilde{\cal H}_{k-1}^{(s,t)}(\kappa^{(t)})/\tilde{\cal H}_{k}^{(s,t)}(\kappa^{(t)})$ with $u_k^{(s,t)}/v_k^{(s,t)}$,
we obtain
\begin{align*}
\kappa^{(t)}=\frac{v_{k+1}^{(s,t)}}{u_{k+1}^{(s,t)}}+u_{k}^{(s,t)},\quad k=0,1,\dots,2m-1,
\end{align*}
which leads to \eqref{vu}.
Moreover, Lemma~\ref{lem12} with $z=\kappa^{(t)}$ yields
\begin{align*}
V_{k}^{(s,t)}=-\frac{\tilde{\cal H}_{k+1}^{(s,t)}(\kappa^{(t)})}{\tilde{\cal H}_{k-1}^{(s,t)}(\kappa^{(t)})}.
\end{align*}
Noting that
\begin{align*} 
\frac{\tilde{\cal H}_{k+1}^{(s,t)}(\kappa^{(t)})}{\tilde{\cal H}_{k-1}^{(s,t)}(\kappa^{(t)})}=\frac{\tilde{\cal H}_{k+1}^{(s,t)}(\kappa^{(t)})}{\tilde{\cal H}_{k}^{(s,t)}(\kappa^{(t)})}\frac{\tilde{\cal H}_{k}^{(s,t)}(\kappa^{(t)})}{\tilde{\cal H}_{k-1}^{(s,t)}(\kappa^{(t)})},
\end{align*}
we derive
\begin{align}
V_k^{(s,t)}=-\frac{v_{k}^{(s,t)}v_{k+1}^{(s,t)}}{u_{k}^{(s,t)}u_{k+1}^{(s,t)}}.\label{Vvvuu}
\end{align}
Combining \eqref{vu} with \eqref{Vvvuu} gives \eqref{Vu}.\par
Equation \eqref{vV} with \eqref{vu} and \eqref{Vu} immediately leads to \eqref{odLV}. 
It is easy to check that $v_k^{(s,t)}$ in \eqref{vu} and $V_k^{(s,t)}$ in \eqref{Vu} also satisfy \eqref{v+V}.

\end{proof}
\par
Without loss of generality, we can fix $s=0$ for examining the dLV system \eqref{dLV} 
because the system does not give the evolution from $s$ to $s+1$.
Thus, the replacements $u_k^{(0,t)}=[1/(\kappa^{(t)})]u_k^{(t)}$
and $(\kappa^{(t)})^2=-1/\delta^{(t)}$ in \eqref{odLV} of Lemma~\ref{lem14}
simply generate the dLV system \eqref{dLV}.
From $e_0^{(s,t)}=0$ and $e_{m}^{(s,t)}=0$,
we also derive the boundary condition $u_0^{(t)}=0$ and $u_{2m}^{(t)}=0$ 
in the dLV system \eqref{dLV}.
\par
Lemma~\ref{lem14} with $v_{2k-1}^{(s,t)}=q_k^{(s,t)}$ and $v_{2k}^{(s,t)}=e_k^{(s,t)}$ suggests
that the dLV variables $u_k^{(t)}$ can be expressed using $q_k^{(s,t)}$ and $e_k^{(s,t)}$.
Recalling that $q_k^{(s,t)}$ and $e_k^{(s,t)}$ are expressed using the Hankel determinants $H_k^{(s,t)}$,
we obtain a theorem concerning the solution to the dLV system \eqref{dLV}.
%
%
\begin{theorem}\label{thm7}
Let us assume that for the infinite sequence $\{f_s^{(t)}\}_{s,t=0}^{\infty}$, 
 $H_k^{(s,t)}$ are all nonzero.
Then, the dLV variables $u_k^{(t)}$ can be expressed using the Hankel determinants $H_k^{(s,t)}$ as
\begin{align}
&u_{2k-1}^{(t)}=\frac{H_k^{(1,t)}H_{k-1}^{(0,t+1)}}{H_k^{(0,t)}H_{k-1}^{(1,t+1)}},\quad k=1,2,\dots,m,\label{u2k-1tau}\\
&u_{2k}^{(t)}=\frac{1}{\delta^{(t)}}
\frac{H_{k+1}^{(0,t)}H_{k-1}^{(1,t+1)}}{H_k^{(1,t)}H_{k}^{(0,t+1)}},\quad k=0,1,\dots,m.
\label{u2ktau}
\end{align}
\end{theorem}
\begin{proof}
By taking into account that 
$u_{2k-1}^{(s,t)}=v_{2k-1}^{(s,t)}(\tilde{\cal H}_{2k-2}^{(s,t)}(\kappa^{(t)})/\tilde{\cal H}_{2k-1}^{(s,t)}(\kappa^{(t)}))$
in Lemma~\ref{lem14} and by using
\eqref{psiphi},
we derive
\begin{align}\label{u2k-1q}
u_{2k-1}^{(s,t)}=q_k^{(s,t)}\frac{{\cal H}_{k-1}^{(s,t)}(\mu^{(t)})}{\kappa^{(t)}{\cal H}_{k-1}^{(s+1,t)}(\mu^{(t)})}.
\end{align}
Since \eqref{Tkzjacobitproof} with $z=\mu^{(t)}$ becomes $H_k^{(s,t)}(\mu^{(t)})=(-1)^kH_k^{(s,t+1)}$, 
we see that ${\cal H}_{k-1}^{(s,t)}(\mu^{(t)})=H_{k-1}^{(s,t)}(\mu^{(t)})/H_{k-1}^{(s,t)}=(-1)^kH_{k-1}^{(s,t+1)}/H_{k-1}^{(s,t)}$.
Thus, by combining it with \eqref{u2k-1q}, we obtain
\begin{align}\label{temp128}
u_{2k-1}^{(s,t)}=\frac{q_k^{(s,t)}}{\kappa^{(t)}}
\frac{H_{k-1}^{(s,t+1)}H_{k-1}^{(s+1,t)}}{H_{k-1}^{(s,t)}H_{k-1}^{(s+1,t+1)}}.
\end{align}
Considering $q_k^{(s,t)}$ in \eqref{qkst}, we can rewrite \eqref{temp128} as
\begin{align*}
&u_{2k-1}^{(s,t)}=\frac{1}{\kappa^{(t)}}
\frac{H_k^{(s+1,t)}H_{k-1}^{(s,t+1)}}{H_k^{(s,t)}H_{k-1}^{(s+1,t+1)}},\quad k=1,2,\dots,m.
\end{align*}
Noting that $u_{2k-1}^{(t)}=\kappa^{(t)}u_{2k-1}^{(0,t)}$, we have \eqref{u2k-1tau}.
\par
Similarly, it follows that
\begin{align*}
u_{2k}^{(s,t)}
&=e_k^{(s,t)}\frac{\kappa^{(t)}{\cal H}_{k-1}^{(s+1,t)}(\mu^{(t)})}{{\cal H}_{k}^{(s,t)}(\mu^{(t)})}\nonumber\\
&=-e_k^{(s,t)}\kappa^{(t)}
\frac{H_{k-1}^{(s+1,t+1)}H_k^{(s,t)}}
{H_{k-1}^{(s+1,t)}H_k^{(s,t+1)}}.
\end{align*}
By combining $e_k^{(s,t)}$ in \eqref{ekst} with this,
we obtain 
\begin{align*}
u_{2k}^{(s,t)}=-\kappa^{(t)}
\frac{H_{k+1}^{(s,t)}H_{k-1}^{(s+1,t+1)}}{H_k^{(s+1,t)}H_{k}^{(s,t+1)}},\quad k=0,1,\dots,m.
\end{align*}
Therefore, the specialization $u_{2k}^{(t)}=\kappa^{(t)}u_{2k}^{(0,t)}$ yields \eqref{u2ktau}.

\end{proof}
\par
Using the substitution $\tilde{H}_k^{(s,t)}:=(\delta^{(0)}\delta^{(1)}\cdots\delta^{(t-1)})^{k}H_k^{(s,t)}$,
the resulting solution is equivalent to that shown in \cite{Iwasaki2002,Spiridonov1997}.
Although the discrete-time variable $s$ does not appear in the dLV system \eqref{dLV},
it plays a key role in expressing the determinant solution to the dLV system \eqref{dLV}.
It is emphasized that the determinant solution to the dLV system \eqref{dLV} is given similar to the cases of the autonomous dToda equation \eqref{adToda} and the non-autonomous dToda equation \eqref{nadToda} from the viewpoint of the infinite sequence $\{f_s^{(t)}\}_{s,t=0}^{\infty}$ involving two types of discrete-time variables $s$ and $t$.
\par
Theorem~\ref{thm7} with Lemma~\ref{lem10} enables us to asymptotically analyze the dLV system \eqref{dLV}.
%
%
\begin{theorem}\label{thm8}
Let us assume that for the infinite sequence $\{f_s^{(t)}\}_{s,t=0}^{\infty}$,  $H_k^{(s,t)}$ are all nonzero.
Moreover, let $\lambda_1,\lambda_2,\dots,\lambda_m$ be constants such that 
$|\lambda_1-\mu^{(t)}|>|\lambda_2-\mu^{(t)}|>\dots>|\lambda_m-\mu^{(t)}|$
and the limit $\delta^{(t)}$ exists.
Then, it holds that 
\begin{align}
&\lim_{t\rightarrow\infty}
u_{2k-1}^{(t)}=\lambda_k,\quad k=1,2,\dots,m,\label{limu2k-1tau}
\\
&\lim_{t\rightarrow\infty}
u_{2k}^{(t)}=0,\quad k=1,2,\dots,m-1.\label{limu2ktau}
\end{align}
\end{theorem}
\begin{proof}
From Lemma~\ref{lem10} and Theorem~\ref{thm7},
we derive
\begin{align*}
&\lim_{t\rightarrow\infty}
u_{2k-1}^{(t)}=
\frac{\check{c}_k^{(1)}\check{c}_{k-1}^{(0)}}{\check{c}_k^{(0)}\check{c}_{k-1}^{(1)}},\\
&\lim_{t\rightarrow\infty}
u_{2k}^{(t)}=0.
\end{align*}
Noting that 
$\check{c}_k^{(1)}\check{c}_{k-1}^{(0)}/(\check{c}_k^{(0)}\check{c}_{k-1}^{(1)})=\lambda_k$,
we immediately have \eqref{limu2k-1tau} and \eqref{limu2ktau}.
\end{proof} 
\par

The convergence theorem concerning the dLV system \eqref{dLV} is restricted in \cite{Iwasaki2002,Tsujimoto2001} to the case where $\delta^{(t)}$ is positive at every $t$. 
In addition, in \cite{Iwasaki2002,Tsujimoto2001}, 
the dLV system \eqref{dLV} is associated with matrix similarity transformations for analyzing the asymptotic behavior and the positivity of the parameter $\delta^{(t)}$ is shown to be a sufficient condition for convergence. 
Without discussing the associated similarity transformations,
Theorem~\ref{thm8}, however, claims that the convergence theorem similar to in \cite{Iwasaki2002,Tsujimoto2001} holds 
even if $\delta^{(t)}$ is negative at each $t$.
Suitable negative $\delta^{(t)}$ in the dLV system \eqref{dLV} realize the introduction of the effective shifts into the similarity transformations for accelerating the convergence \cite{Iwasaki2004,Yamamoto2010}.
Theorem~\ref{thm8} thus shows the convergence of a shifted algorithm for computing singular values of bidiagonal matrices, 
which is designed based on the dLV system \eqref{dLV}.

\section{Conclusion}\label{sec:7}
The main objective of this paper is to describe and understand the well-known integrable discrete systems, the discrete Toda equations (dToda) and the discrete Lotka-Volterra (dLV) system, in terms of their determinant solutions and asymptotic behavior. 
The key point is to introduce an infinite sequence with respect to two types of discrete-time variables.
\par
We first examined properties of the Hankel determinants and the Hadamard polynomials associated with the infinite sequence. 
Then, we showed a Hankel determinant expression of solutions with sufficient degrees of freedom for the autonomous dToda equation with no parameter and the non-autonomous dToda equation with parameters. 
Next, by using asymptotic expansions of the Hankel determinants, we presented asymptotic analysis of the autonomous and non-autonomous dToda equations as the discrete-time variables go to infinity. 
We finally clarified the determinant solution to the dLV system with parameters, and observed its asymptotic behavior as the discrete-time variable goes to infinity. 
In particular, asymptotic analysis of the non-autonomous dToda equation and the dLV system concluded the asymptotic convergence of the shifted qd and dLV algorithms for computing singular values of bidiagonal matrices providing a regularity condition on the Hankel determinants.


\begin{thebibliography}{99}
\bibitem{Bogoyavlensky1976}
Bogoyavlensky, O.I.:
On Perturbations of the Periodic Toda Lattice,
Commun. math. Phys.,
{\bf 51}, 201--209 (1976)

\bibitem{Chihara}
Chihara, T.S.:
An introduction to orthogonal polynomials,
Gordon and Breach Science Publishers,
New York--London--Paris (1978)

\bibitem{Henrici}
Henrici, P.:
Applied and Computational Complex Analysis,
vol.1.John Wiley, New York (1974)

\bibitem{Hirota1973}
Hirota, R.:
Exact N-soliton solution of a nonlinear lumped network equation,
J. Phys. Soc. Japan,
{\bf 35}, 289--294 (1973)

%

\bibitem{Hirota1997}
Hirota, R.:
Conserved Quantities of ``Random-Time Toda equation",
J. Phys. soc. Jpn.,
{\bf 66}, 283--284 (1997)

\bibitem{Hirota2003}
Hirota, R.:
Determinants and Pfaffians, 
surikaisekikenkyusho kokyuroku,
{\bf 1302}, 220--242 (2003)

\bibitem{Iwasaki2002}
Iwasaki, M., Nakamura, Y.:
On the convergence of a solution of the discrete Lotka-Volterra system,
Inverse Plob.,
{\bf 18}, 1569--1578 (2002)

\bibitem{Iwasaki2004}
Iwasaki, M., Nakamura, Y.:
An application of the discrete Lotka-Volterra system with variable step-size to singular value computation,
Inverse Plob.,
{\bf 20}, 553--563 (2004)


\bibitem{Maeda2013}
Maeda, K., Tsujimoto, S.:
Direct Connection between the RII Chain and the Nonautonomous Discrete Modified KdV Lattice,
SIGMA. 
{\bf 9}, 073 (2013)


\bibitem{Nagai1998}
Nagai, A., Tokihiro, T., Satsuma, J.:
The Toda molecule equation and the $\epsilon$-algorithm,
Appl.Math.Comput.,
{\bf 67}, 1565--1575 (1998)

\bibitem{Nakamura1997}
Nakamura. A.:
Explicit $N$-soliton solutions of the $1+1$ dimensional Toda molecule equation,
 J. Phys. Soc. Jpn., 
 {\bf 67}, 791--798 (1998)

\bibitem{Parlett1995}
Parlett, B.N.:
The new qd algorithm,
Acta Numer.,
{\bf 4}, 459--491 (1995)

\bibitem{Rutishauser1990}
Rutishauser, H.:
Lectures on numerical mathematics.
Birh\"auser,
Boston (1990)

\bibitem{Spiridonov1997}
Spiridonov, V., Zhedanov, A.: 
Discrete-time Volterra chain and classical orthgonal polynomial,
J. Phys. A: Math. Gen.,
{\bf 30}, 8727--8737 (1997)

\bibitem{Symes1982}
Symes, W. W.:
The $QR$ algorithm and scattering for the finite nonperiodic Toda lattice, 
Physica, {\bf 4D}, 275--280 (1982)

\bibitem{Toda1981}
Toda, M.:
Theory of Nonlinear Lattices,
Springer-Verlag, NewYork (1981)

\bibitem{Tsujimoto2001}
Tsujimoto, S., Nakamura, Y., Iwasaki, M.:
The discrete Lotka-Volterra system computes singular values,
Inverse Plob.,
{\bf 17}, 51--58 (2001)

\bibitem{Yamamoto2010}
Yamamoto, Y., Fukuda, A., Iwasaki, M., Ishiwata, E., Nakamura, Y.: 
On a variable transformation between two integrable systems: The discrete hungry Toda equation and the discrete hungry Lotka-Volterra system,
AIP Conf. Proc.,
{\bf 1281}, 2045--2048 (2010)

\bibitem{Zhedanov}
Zhedanov, A.:
Rational spectral transformations and orthogonal polynomials,
J. Comput. Appl. Math.,
{\bf 85}, 67--86 (1997)


\end{thebibliography}
\end{document}